\documentclass[11pt]{llncs}

\usepackage{comment}
\usepackage{boxedminipage}
\usepackage{a4wide}
\usepackage{ntabbing}
\usepackage{enumerate}
\usepackage{amssymb}
\usepackage{graphicx}
\usepackage{color}

\newtheorem{thm}{Theorem}
\newtheorem{lem}[thm]{Lemma}
\newtheorem{corol}[thm]{Corollary}
\newtheorem{propo}[thm]{Proposition}

\def\C{\mathcal{C}}
\def\TTT{\mathcal{T}}

\def\restr{|}
\def\bs{\setminus}
\def\SSS{\mathcal{S}}
\def\CSG{CSG}

\title{Using Contracted Solution Graphs for\\ Solving Reconfiguration Problems\thanks{An extended abstract of this paper appeared in the proceedings of MFCS 2016~\cite{BP16}.}}

\author{Paul Bonsma\thanks{Supported by the European Community's Seventh Framework Programme (FP7/2007-2013), grant agreement n$^{\circ}$ 317662.} \inst{1}
\and
Dani\"el Paulusma\thanks{Supported by EPSRC Grant EP/K025090/1.} \inst{2}
}

\institute{
Saxion University of Applied Sciences, Enschede, The Netherlands,
\texttt{p.s.bonsma@saxion.nl}
\and
Durham University, UK,
\texttt{daniel.paulusma@durham.ac.uk}
}

\begin{document}

\maketitle

\begin{abstract}
We introduce in a general setting a dynamic programming method for solving reconfiguration problems.
Our method is based on {\em contracted solution graphs}, which are obtained from  solution graphs 
by performing an appropriate series of edge contractions that decrease the graph size
without losing any critical information needed to solve the reconfiguration problem under consideration. 
Our general framework captures the approach behind known reconfiguration results of Bonsma (2012) and Hatanaka, Ito and Zhou (2014).
As a third example, we apply  the method to the following well-studied problem:  
given two $k$-colorings $\alpha$ and $\beta$ of a graph~$G$, can $\alpha$ be modified into $\beta$ by recoloring one vertex of $G$ at a time, while maintaining a $k$-coloring throughout?
This problem is known to be PSPACE-hard even for bipartite planar graphs and $k=4$. 
By applying our method in combination with a thorough exploitation of the graph structure we obtain a polynomial-time algorithm for $(k-2)$-connected chordal graphs. 

\end{abstract}

\section{Introduction}
\label{sec:intro}
 
Solving a given instance of an {\sf NP}-hard search problem means that we need to explore an exponentially large solution space. In order to get more insight into the solution space, it is a natural question to check how ``close'''
one solution for a particular instance is to another solution of that instance.
Doing so could, for instance, be 
potentially interesting for improving the performance of corresponding heuristics~\cite{GKMP09}. Searching the solution space by making small ``feasible'' moves also turned out to be useful when analyzing randomized algorithms for sampling and counting $k$-colorings of a graph or when analyzing cases of Glauber dynamics in statistical physics (see Section~5 of the survey of van den Heuvel~\cite{He13}).
Solution spaces in practical problems, such as stacking problems arising in storage planning~\cite{KLMSS07}, have been explored in a similar matter.

The above situation can be modeled as follows. 
A {\em solution graph concept} $\SSS$ is obtained by defining a set of {\em instances}, {\em solutions} for these instances, and a (symmetric) {\em adjacency relation} between pairs of solutions. For every instance $G$ of the problem, this gives a {\em solution graph}~$\SSS(G)$, also called a {\em reconfiguration graph}, which has as node 
set all solutions of $G$, with edges as defined by some given adjacency relation (if $G$ has no solutions, then $\SSS(G)$ is the empty graph).
The adjacency relation usually represents a smallest possible change ({\em reconfiguration move}) between two solutions of the same instance. 
For example, the well-known
{\em $k$-Color Graph concept} $\C_k$, related to the {\sc $k$-Coloring} search problem, is defined as follows: instances are graphs $G$, and solutions are (proper) $k$-colorings of $G$. Two $k$-colorings are adjacent nodes in
the reconfiguration graph $\C_k(G)$
 if 
and only if they differ in exactly one vertex of $G$.
In general there may be more than one natural way to define the adjacency relation.

Solution graphs and their properties have been studied very intensively over the last couple of years for a variety of search problems, such as {\sc $k$-Coloring}~\cite{BB18,BDFJP17,BJLPP14,BC09,CHJ08,CHJ09,CHJ11,FJP16,JKKPP16},
{\sc Satisfiability}~\cite{GKMP09,MNPR17},
{\sc Independent Set}~\cite{Bo16,BKW14,KMM12}, 
{\sc Shortest Path}~\cite{Bo17,Bo12B,KMM12}, 
{\sc List Coloring}~\cite{HIZ15},
{\sc List Edge Coloring}~\cite{IKD12,IKZ12},
{\sc $L(2,1)$-Labeling}~\cite{IKOZ14},
{\sc $H$-Coloring}~\cite{Wr15}
and {\sc Subset Sum}~\cite{ID14}; see also the surveys~\cite{He13,Ni18}.
The study of such solution graphs is commonly called {\em reconfiguration}. 

The area of reconfiguration is fast growing, and
both algorithmic and combinatorial questions have been considered.
For instance, what is the diameter
of $\SSS(G)$ (in terms of the size of the instance $G$) or if $\SSS(G)$ is not connected, what is the diameter of its (connected) components? In particular, is the diameter always polynomially bounded or not?
This led to the introduction of 
the {\sc $\SSS$-Connectivity} problem, which is that of deciding 
whether the solution graph $\SSS(G)$ of a given instance~$G$ is connected. 
We consider the following related problem, which is also a central problem in the area of reconfiguration:

\begin{center}
\begin{boxedminipage}{.99\textwidth}
\textsc{$\SSS$-Reachability}\\[2pt]
\begin{tabular}{ r p{0.8\textwidth}}
\textit{~~~~Instance:} &an instance $G$ with two solutions $\alpha$ and $\beta$.\\
\textit{Question:} & is there a path from $\alpha$ to $\beta$ in $\SSS(G)$?
\end{tabular}
\end{boxedminipage}
\end{center}

\smallskip
\noindent
The {\sc $\SSS$-Reachability} problem is sometimes called the {\em $\alpha$-$\beta$-path} problem for $\SSS$~\cite{He13}.  Our example problem {\sc $\C_k$-Reachability} is also known as the {\sc $k$-Color Path} problem~\cite{CHJ11}.

\begin{center}
\begin{boxedminipage}{.99\textwidth}
\textsc{$\C_k$-Reachability}\\[2pt]
\begin{tabular}{ r p{0.8\textwidth}}
\textit{~~~~Instance:} &an instance $G$ with two $k$-colorings $\alpha$ and $\beta$.\\
\textit{Question:} & is there a path from $\alpha$ to $\beta$ in $\C_k(G)$?
\end{tabular}
\end{boxedminipage}
\end{center}

It is known 
that {\sc $\SSS$-Reachability} is PSPACE-complete for most of the aforementioned solution graph concepts
even for 
special classes of instances~\cite{BC09,HIMNOST16,IDHPSUU11,MNRW14,Wr18,Za15}.
For instance, {\sc $\C_k$-Reachability} is PSPACE-complete even if $k=4$ and instances are restricted to planar bipartite graphs~\cite{BC09}. 
This explains that
efficient algorithms are only known
for very restricted 
classes of instances.
Hence, there is still a need for developing general algorithmic techniques
for solving these problems in practice, and 
for sharpening the boundary between tractable and computationally hard instance classes.

One important algorithmic technique is dynamic programming (DP).
There are only relatively few successful examples of nontrivial dynamic programming algorithms for solving {\sc $\SSS$-Reachability} problems.
The reason for this is that many well-studied {\sc $\SSS$-Reachability} problems (including {\sc $\C_k$-Reachability} for an appropriate constant $k$) are PSPACE-complete even for graphs of bounded bandwidth~\cite{MNRW14,Wr18}, and therefore also for graphs of bounded treewidth.
In fact, the PSPACE-completeness results from~\cite{MNRW14,Wr18} hold even for planar graphs of bounded bandwidth and low maximum degree~\cite{Za15}.

One way to cope with the above problem is to restrict the problem even further. 
For instance, in a number of recent papers~\cite{BMNR14,HIZ18,JKKPP16,KMM12,MNPR17,MNRS18,MNRSS17} the {\it length-bounded} version of the {\sc $\SSS$-Reachability} problem was studied. This is the problem of finding a path of length at most~$\ell$ in the solution graph between two given solutions. 
Taking the length~$\ell$ of a path between two solutions as a natural parameter, a particular aim of these papers was to determine fixed-parameter tractability.
For instance, although {\sc $\C_k$-Reachability} is PSPACE-complete for $k\ge 4$, the length-bounded version is FPT when parameterized by the length~$\ell$~\cite{BMNR14,JKKPP16} (for $k\leq 3$, the length-bounded version is even polynomial-time solvable~\cite{JKKPP16}).
In this restricted context, dynamic programming algorithms over tree decompositions for reconfiguration problems are more common.
For instance, in~\cite{MNRW14}
FPT algorithms are given for various length-bounded reachability problems, parameterized by both the treewidth and the length~$\ell$. To give another example, in~\cite{LMPRS18} FPT algorithms are given for the reachability versions of different token reconfiguration problems for graphs of bounded degeneracy (and thus for bounded treewidth), when parameterized by the number of tokens.

\subsection{Aims and Methodology}

We aim to solve the (original) {\sc $\SSS$-Reachability} problem via an algorithm that uses a generally applicable  DP method based on {\em contracted solution graphs}. Due to the PSPACE-completeness
of  {\sc $\SSS$-Reachability}, such an algorithm does not terminate in polynomial time for all instances. Hence we aim to identify restricted instance classes for which we do obtain a polynomial running time, as illustrated by a new application explained in Section~\ref{s-appp} and two known examples~\cite{Bo17,HIZ15} explained below.  

Bonsma~\cite{Bo17} introduced the DP method based on contracted solution graphs to obtain an efficient algorithm for {\sc Shortest-Path-Reachability}  restricted to planar graphs. 
Hatanaka, Ito and Zhou~\cite{HIZ15} used 
this DP method for proving that {\sc List-Coloring-Reachability} is polynomial-time solvable for caterpillars
(in both papers, contracted solution graphs are called {\em encodings}).
To be more precise, in~\cite{HIZ15} dynamic programming
was done over a {\em path decomposition} of the given caterpillar. In~\cite{Bo17}, a layer-based decomposition of the graph was used, which can also be viewed as a path decomposition. 
In our paper we focus on the more general {\em tree decompositions} instead (which requires us 
to introduce a join rule).
We will generalize the ideas of~\cite{Bo17,HIZ15} to a unified DP method
and illustrate the method by giving a new application.

We now sketch our method and refer to Section~\ref{s-method} for a detailed description.
In dynamic programming
one first computes the required information for
parts of an instance~$G$.
One combines/propagates this to compute the same information for ever larger parts of the instance, until the desired information is known for $G$ entirely. In our case, $G$ can be any relational structure on a ground set, such as (directed) graphs, hypergraphs, satisfiability formulas, or constraint satisfaction problems in general (see e.g.~\cite{BMNR14}). 
The order in which the information can be computed or the order in which parts must be considered is given by a {\em decomposition} of $G$. 
For a processed part $H$ of $G$, the elements of the ground set that are  in~$H$
and that
have incidences with the unexplored part
are called {\em terminals}.
Reconfiguration moves in $H$ that do not involve terminals are often irrelevant. We capture the information that is relevant  by the notion of a {\em terminal projection}. 
These projections assign labels to solutions, yielding so-called {\em label components}, which are maximally connected subgraphs of $\SSS(H)$ induced by sets of solutions that all have the same label. A {\it contracted solution graph} is
obtained from $\SSS(H)$ by contracting the label components into single vertices. 
Dynamic programming rules for a given decomposition of $G$ describe how to compute new (larger) contracted solution graphs from smaller ones.
 
In Section~\ref{sec:DP} we illustrate 
our method by giving dynamic programming rules for the {\sc $\C_k$-Reachability} 
problem that can be used if a  tree decomposition of the graph is given.
Recall that similar dynamic programming rules have been for other reconfiguration problems~\cite{Bo17,HIZ15} when a path decomposition is given.
Our rules solve the {\sc $\C_k$-Reachability problem} 
correctly for every graph $G$ and can also be used directly for  {\sc List-Coloring-Reachability} and thus generalize the rules of~\cite{HIZ15}.
Nevertheless, the algorithm is only {\em efficient} when the contracted solution graphs stay small enough 
(that is, polynomially bounded).
As indicated
by the aforementioned PSPACE-hardness of  {\sc $\C_k$-Reachability}, this is 
not always the case. 
To make this explicitly clear,
in Section~\ref{sec:badexamples}, 
we apply the DP rules on a specific example, which shows that
the size of the contracted solution graphs can indeed grow exponentially, even for 2-connected 4-colorable unit interval graphs. 

\subsection{Application}\label{s-appp}

In Section~\ref{sec:chordalgraphs} we apply the DP method to show that, for all $k\geq 3$, {\sc $\C_k$-Reachability}  
is polynomial-time solvable 
for $(k-2)$-connected chordal graphs.
Chordal graphs form a well-studied graph class; see e.g.~\cite{BLS99} for more information.
As unit interval graphs are chordal, the  example given in Section~\ref{sec:badexamples} 
implies that we need to
exploit the structure of chordal graphs  further in combination with
applying the DP method. The idea is to show that it suffices to compute the contracted solution graphs only partly.
In order to do this we introduce the new notion of injective neighbourhood property of contracted solution graphs for
 {\sc $\C_k$-Reachability}, which helps us to characterize contracted solution graphs if the original graph~$G$ is chordal and $(k-2)$-connected.

As the proof for the PSPACE-completeness result for bipartite graphs from~\cite{BC09} can be easily modified to hold for $(k-2)$-connected bipartite graphs, our result
for $(k-2)$-connected chordal graphs cannot be extended to $(k-2)$-connected perfect graphs.
Our result cannot be extended to all chordal graphs either: recently, Hatanaka, Ito and Zhou~\cite{HIZ17b} solved an open problem  posed in the conference version of our paper~\cite{BP16} by proving that {\sc $\C_k$-Reachability} is PSPACE-complete for chordal graphs if $k$ is a sufficiently large constant.
We note that in contrast, {\sc $\C_k$-Connectivity} is polynomial-time solvable on chordal graphs. This is due to a 
more general 
result of Bonamy et al.~\cite{BJLPP14}, which implies that for a chordal graph $G$, $\C_k(G)$ is connected if and only if $G$ has no clique with more than $k-1$ vertices. 

Our result on {\sc $\C_k$-Reachability} on $(k-2)$-connected chordal graphs is the 
first time that dynamic programming over tree decompositions is used to solve 
the general version of a PSPACE-complete reachability problem in polynomial time for a graph class strictly broader than trees.
As reachability problems become quickly PSPACE-complete, we can only hope to obtain polynomial-time algorithms for rather restricted graph classes. Nevertheless we believe that there are a number of interesting open problems in this direction, for which our method could be useful (see Section~\ref{s-discussion}). We also 
remark that the true strength of our method is not always revealed when using the viewpoint of worst-case algorithm analysis. 
For instance, we observed from some initial experiments using randomly generated $k$-colorable chordal or interval graphs that
the method performs well on most instances, despite the fact that specialized examples can be constructed that exhibit exponential growth. Discussing this is beyond the scope of the current paper. However, being able to use our method for performing experiments with a goal to further refining it was one of the reasons why we choose to present it in full generality.

\section{Preliminaries}\label{s-pre}

We consider finite undirected graphs that have no multi-edges and no loops. 
Below we define some basic terminology. 
In particular we give some coloring terminology, 
as we need such terminology throughout the paper.
We refer to the textbook of Diestel~\cite{Di10} for any undefined terms.

For a connected graph $G$, a {\em vertex cut} is a set $S\subseteq V(G)$ such that $G-S$ is disconnected.
Vertices in different components of $G-S$ are said to be {\em separated} by $S$.
For $k\ge 1$, a (connected) graph $G$ is {\em $k$-connected} if $|V(G)|\ge k+1$ and every vertex cut $S$ has $|S|\ge k$. The \emph{contraction} of an edge~$uv$ of a graph $G$  replaces $u$ and $v$ by 
 a new vertex made adjacent to precisely those vertices that were adjacent to $u$ or $v$ in $G$
(this does not create any multi-edges or loops).
A graph is {\em chordal} if it has no induced cycle of length greater than~3. 

Let $G$ be a graph.
A {\em $k$-color assignment} of  $G$ is a function $\alpha:V(G)\to \{1,\ldots,k\}$. For $v\in V(G)$,  $\alpha(v)$ is called the {\em color} of $v$. 
A $k$-color assignment~$\alpha$ is a {\em $k$-coloring} if  $\alpha(u)\not=\alpha(v)$ for every edge $uv\in E(G)$. 
A {\em coloring} of $G$ is a $k$-coloring for some value of $k$. 
If  $\alpha$ and $\beta$ are colorings of $G$ and a subgraph $H$ of $G$, respectively, such that
$\alpha\restr_{V(H)}=\beta$ (that is, 
$\alpha$ and $\beta$
coincide on $V(H)$)  then  
$\alpha$ and $\beta$ are said to be {\em compatible}.

For an
integer $k$, the {\em $k$-color graph $\C_k(G)$} has as nodes
all (proper) $k$-colorings of $G$, such that two colorings are adjacent if and only if they differ on one vertex.
A {\em walk} from $u$ to $v$ in $G$ is a sequence of vertices $v_0,\ldots,v_k$ with $u=v_0$, $v=v_k$, such that for all $i<k$, $v_iv_{i+1}\in E(G)$. A {\em pseudowalk} from $u$ to $v$ is a sequence of vertices $v_0,\ldots,v_k$ with $u=v_0$, $v=v_k$, such that for all $i<k$, either $v_i=v_{i+1}$, or $v_iv_{i+1}\in E(G)$.
A {\em recoloring sequence} from a $k$-coloring $\alpha$ of $G$ to a $k$-coloring $\beta$ of $G$ is a pseudowalk from $\alpha$ to $\beta$ in $\C_k(G)$.

A {\em labeled graph} is a pair $(G,\ell)$ where $G=(V,E)$ is a graph and $\ell:V\to X$ for some set~$X$ is 
a vertex labeling (which may assign the same label to different vertices); we refer to Section~\ref{s-method} for an example.
A  {\em label preserving isomorphism} between two labeled graphs $(G_1,\ell_1)$ and $(G_2,\ell_2)$ is an isomorphism $\phi:V(G_1)\to V(G_2)$, such that  $\ell_1(v)=\ell_2(\phi(v))$ for all $v\in V(G_1)$. 
Informally, two labeled graphs $(G_1,\ell_1)$ and $(G_2,\ell_2)$ 
are the same if there exists a label preserving isomorphism between them. 

\section{The Method of Contracted Solution Graphs}\label{s-method}

In this section we define the concept of {\it contracted solution graphs} (CSGs)
for reconfiguration problems in general.
Consider a solution graph concept $\SSS$, 
which for every instance $G$ of $\SSS$
defines a solution graph that is denoted by $\SSS(G)$. 
A {\em terminal projection} for $\SSS$ is a 
function~$p$ that 
assigns a {\it label} to each tuple $(G,T,\gamma)$ consisting of an instance $G$ of $\SSS$, a set $T$ of {\em terminals} for $G$ and 
a solution $\gamma$ for $G$.
Terminal projections are used to decide which nodes are ``equivalent'' and can be contracted.
In our example and in previous examples in the literature~\cite{Bo17,HIZ15} $G$ is always a graph, and $T$ is a subset of its vertices.
A terminal projection~$p$  can be seen as a node labeling for the solution graph $\SSS(G)$. So, for every instance  $G$ of $\SSS$, every choice of terminals $T$ may give a different 
node labeling for the solution graph $\SSS(G)$. 
If $G$ and $T$ are clear from the context,  we write $p(\gamma)$ to denote the label of a node $\gamma$ of $\SSS(G)$. 

\medskip
\noindent
{\it Example 1.} Consider the $k$-color graph concept $\C_k$. Let $G$ be a graph. We can define a terminal projection $p$ as follows.
Let $T$ be a subset of $V(G)$. The nodes of $\C_k(G)$ are $k$-colorings and we give each node as label its restriction to $T$, that is, for every $k$-coloring $\gamma$ of $G$, we set
$p(\gamma)=p(G,T,\gamma)=\gamma\restr_{T}$. Note that 
$\gamma\restr_{T}$ is a $k$-coloring of $G[T]$.

\medskip
\noindent
Let $p$ be a terminal 
projection
for a solution graph concept $\SSS$.
For an 
instance~$G$ of $\SSS$ and a terminal set~$T$, a 
{\em label component} $C$ of $\SSS(G)$ is a maximal set of nodes $\gamma$ that all have the same label $p(\gamma)$ and that induce a connected subgraph of $\SSS(G)$.
It is easy to see that every solution~$\gamma$ of~$G$ is part of exactly one label component, or in other words: the label components partition the node set of $\SSS(G)$.
The {\em contracted solution graph (CSG)} $\SSS^c(G,T)$ is a labeled graph that 
has a node set that corresponds bijectively to the set of label components of~$G$.
For a node $x$ of $\SSS^c(G,T)$, we denote by $S_x$ the corresponding label component.
Two distinct nodes $x_1$ and $x_2$ of $\SSS^c(G,T)$ are adjacent if and only if there exist solutions $\gamma_1\in S_{x_1}$ and $\gamma_2\in S_{x_2}$ such that $\gamma_1$ and $\gamma_2$ are adjacent in $\SSS(G)$. 
We define a label function $\ell^*$ for nodes of $\SSS^c(G,T)$ to denote the corresponding label in $\SSS(G)$. More precisely: for a node $x$ of $\SSS^c(G,T)$, 
the label $\ell^*(x)$ is chosen such that $\ell^*(x)=p(\gamma)$ for all $\gamma\in S_x$.
Note that the contracted solution graph~$\SSS^c(G,T)$ can also be obtained from $\SSS(G)$ by contracting all label components into single nodes and choosing node labels appropriately.

\medskip
\noindent
{\it Example 2.} 
Figure~\ref{fig:cutvertCSG}(c) shows one component of $\C_4(G)$ for the 
(4-colorable) 
graph $G$ from Figure~\ref{fig:cutvertCSG}(a). This is the component that contains all colorings of~$G$ 
whose vertices $a,b,c,d$ are colored with colors $4,3,2,1$, respectively 
(note that it is not possible to recolor any of these four vertices, as one may recolor only one vertex at a time).
So in Figure~\ref{fig:cutvertCSG}(c) the colors of the vertices $a,b,c,d$ 
are omitted in the node labels, which only indicate the colors of $e,f,g$, in this order. For terminal set $T=\{f\}$, this component contains three label components (of equal size), and contracting them yields the CSG $\C^c_4(G,\{f\})$ shown in Figure~\ref{fig:cutvertCSG}(d). For $T=\{g\}$, there are seven label components, and the corresponding CSG $\C^c_4(G,\{g\})$ is shown in Figure~\ref{fig:cutvertCSG}(e). Note that 
$\C^c_4(G,\{g\})$
contains different nodes with the same label. 
\begin{figure}
\centering
\scalebox{0.5}{$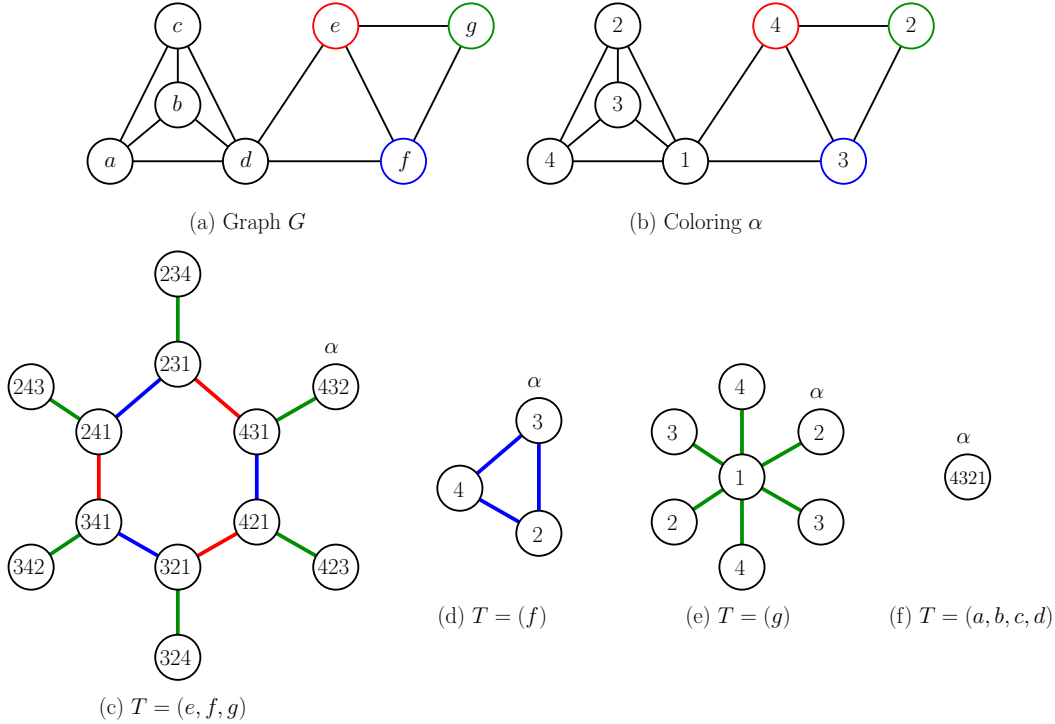$}
\caption{(a) A 4-colorable
 chordal graph $G$
 with $V(G)=\{a,b,c,d,e,f,g\}$.
(b) a 
4-coloring~$\alpha$, and one component of 
the
CSGs of $G$ 
for four 
different
terminal sets $T$: (c) $\C^c_4(G,\{e,f,g\})$, (d) $\C^c_4(G,\{f\})$, (e) $\C^c_4(G,\{g\})$
and (f) $\C^c_4(G,\{a,b,c,d\})$. 
The $G[T]$-colorings in the node labels are given as sequences of colors, for the (ordered version of) $T$ 
as indicated below each CSG.
Example (c) can also be seen as the component of $\C_4(G)$ where vertices $a,b,c,d$ receive colors $4,3,2,1$. 
}
\label{fig:cutvertCSG}
\end{figure}

\medskip
\noindent
We stress that the CSG $\SSS^c(G,T)$ is a labeled graph that includes the label function~$\ell^*$ defined above. However, to keep its size reasonable, the CSG itself does not include the solution sets $S_x$ for each node that were used to define it.
For proving the correctness of 
dynamic programming rules
for CSGs the following alternative characterization of CSGs 
is useful; note that the sets~$S_x$ correspond exactly to the label components. 

\begin{lem}
\label{lem:characterizationCSGs}
Consider an instance $G$ of a solution graph concept $\SSS$, terminal 
set~$T$ and terminal projection $p$. 
Let $(H,\ell)$ be a labeled graph.
Then $(H,\ell) = \SSS^c(G,T)$ if and only if one can define (nonempty) sets of solutions $S_x$ for each node $x\in V(H)$ such that the following properties hold:
\begin{enumerate}
 \item\label{pr:partition}
 $\{S_x \mid x\in V(H)\}$ is a partition of the nodes of $\SSS(G)$ (the solutions of $G$).
 \item\label{pr:correctlabels}
 For every $x\in V(H)$ and every solution $\gamma\in S_x$: $p(G,T,\gamma)=\ell(x)$.
 \item\label{pr:propercoloring}
 For every edge 
 $xy\in E(H)$: $\ell(x)\not=\ell(y)$.
 \item\label{pr:connectedsets}
 For every $x\in V(H)$: $S_x$ induces a connected subgraph of $\SSS(G)$.
 \item\label{pr:adjacency}
 For every pair of distinct nodes $x,y\in V(H)$: 
 $xy\in E(H)$ if and only if there exist solutions $\alpha\in S_x$ and $\beta\in S_y$ such that $\alpha$ and $\beta$ are adjacent in $\SSS(G)$.
\end{enumerate}
\end{lem}

\begin{proof}
$\Rightarrow$:
We choose the sets $S_x$ to be the label components, as chosen in the above definition of $\SSS^c(G,T)$. Then the Properties~(\ref{pr:partition}), (\ref{pr:correctlabels}), (\ref{pr:connectedsets}) and~(\ref{pr:adjacency}) follow immediately from the definitions. For Property~(\ref{pr:propercoloring}), we use that label components are {\em maximal} connected node sets with the same label, together with Properties~(\ref{pr:correctlabels}) and~(\ref{pr:adjacency}).

\medskip
\noindent
$\Leftarrow$: Let $(H,\ell)$ be a labeled graph for which solution sets $S_x$ can be defined such that the five properties hold. Consider a node $x\in V(H)$. 
Properties~(\ref{pr:connectedsets}) and~(\ref{pr:correctlabels}) show that all solutions in $S_x$ are part of the same label component; denote this label component by $C$. 
Note that for all $\alpha\in S_x$, $p(G,T,\alpha)=\ell(x)$ due to Property~(\ref{pr:correctlabels}).

In order to show that in fact $S_x=C$, let $\gamma\notin S_x$ be adjacent to some solution in $S_x$. 
We must show that $p(G,T,\gamma)\neq \ell(x)$. 
By Property~(\ref{pr:partition}), $\gamma$ belongs to a set~$S_y$ for some $y\neq x$. By Property~(\ref{pr:adjacency}), we obtain $xy\in E(C)$. Then, by Property~(\ref{pr:propercoloring}), we find that $\ell(y)\neq \ell(x)$. 
Hence $p(G,T,\gamma)=\ell(y)\neq \ell(x)$ due to Property~(\ref{pr:correctlabels}).

Because $S_x$ induces a label component for every $x$,
there exists a bijection $\phi$ between the nodes of $H$ and the label components of $\SSS(G)$ (This is a bijection because of Property~(\ref{pr:partition})). This yields a bijection $\phi'$ between the nodes of $H$ and the nodes of $\SSS^c(G,T)$, which is label preserving by Property~(\ref{pr:correctlabels}) and an isomorphism by Property~(\ref{pr:adjacency}) and the definition of $\SSS^c(G,T)$. Hence $(H,\ell)=\SSS^c(G,T).$\footnote{Since node names are irrelevant, we will simply write $(H,\ell)=\SSS^c(G,T)$ to denote that there is a label preserving isomorphism between the two. More formally, $\SSS^c(G,T)$ can be seen as a class of labeled graphs that are equivalent under labeled isomorphisms.}\qed
\end{proof}

A mapping $S$ that assigns solution sets (or label components) $S_x$ to each node~$x$ of $\SSS^c(G,T)$ that satisfies the properties given in 
Lemma~\ref{lem:characterizationCSGs}
is called a {\em certificate} for $\SSS^c(G,T)$.
Given such a certificate $S$ and a solution $\gamma$ for $G$, we define 
the {\em $\gamma$-node} of $\SSS^c(G,T)$ with respect to $S$
to be the node $x$ with $\gamma\in S_x$. 
For readability, we will not always explicitly mention this certificate when talking about $\gamma$-nodes in $\SSS^c(G,T)$ (except in Lemma~\ref{propo:one} below), but the reader should keep the following convention in mind: when $\gamma$-nodes are identified in $\SSS^c(G,T)$ for multiple solutions $\gamma$, {\em these are all chosen with respect to the same certificate}.

\medskip
\noindent
{\it Example 3.}
In Figures~\ref{fig:cutvertCSG}(c)--(f),
the $\alpha$-node for the coloring $\alpha$ shown in Figure~\ref{fig:cutvertCSG}(b) is marked.
In particular consider $\C^c_4(G,\{g\})$ in Figure~\ref{fig:cutvertCSG}(e).
Since the certificate for $\C^c_4(G,\{g\})$ is not actually indicated in the figure,
the other leaf with label~2 
can also be chosen as the $\alpha$-node (considering the nontrivial label-preserving automorphisms of the graph). Similarly, if we choose a coloring $\beta$ that coincides with $\alpha$ except on nodes $e$ and $f$, where we choose $\beta(e)=3$ 
and $\beta(f)=4$, then
the same two leaves (the ones with label 2) of
$\C^c_4(G,\{g\})$
can be chosen as the $\beta$-node. 
Nevertheless, if both an $\alpha$-node and $\beta$-node are marked, then this will only be correct according to the above convention when they are distinct!
\footnote{It is possible for the given example to choose two solutions $\alpha$ and $\beta$, and correctly mark an $\alpha$-node $x$ with respect to a one certificate $S^1$, and a $\beta$-node $y$ with respect to another certificate $S^2$, such that $\alpha$ and $\beta$ are in different components of $\SSS(G)$, but $x$ and $y$ are in the same component of $\SSS^c(G,T)$. This is clearly not desirable; see Lemma~\ref{propo:one}.}

\medskip
\noindent
The main purpose of 
our 
definitions 
is the 
following 
key lemma.
\begin{lem}
\label{propo:one}
Let 
$(G,T)$
be an instance of a solution graph concept~$\SSS$.
Let $\SSS^c(G,T)$ be the contracted solution graph for some terminal projection~$p$.
Let $\alpha$ and $\beta$ be two solutions and let $x$ and $y$ be the $\alpha$-node resp.\ $\beta$-node with respect to some certificate~$S$.
Then there is a path from $\alpha$ to $\beta$ in $\SSS(G)$ if and only if there is a path from $x$ to $y$ in $\SSS^c(G,T)$.
\end{lem}

\begin{proof}
First suppose that there exists  
a path $\gamma_0,\ldots,\gamma_k$ from $\alpha$ to $\beta$ in $\SSS(G)$. 
Replace every 
solution~$\gamma_i$ in this sequence by the node $v$ of $\SSS^c(G,T)$ with $\gamma_i\in S_v$. By definition, the resulting node sequence starts in $x$, and terminates in $y$. By  Lemma~\ref{lem:characterizationCSGs}(\ref{pr:adjacency}), consecutive nodes in this sequence are the same or adjacent, so this sequence is a pseudowalk from $x$ to $y$. This 
immediately
yields a path from $x$ to~$y$. 

For the other direction, consider a path $v_0,\ldots,v_k$ from $x$ to $y$ in $\SSS^c(G,T)$. For every node $v_i$, $S_{v_i}$ induces a connected subgraph of $\SSS(G)$ (Lemma~\ref{lem:characterizationCSGs}(\ref{pr:connectedsets})). For any two consecutive nodes $v_i$ and $v_{i+1}$, there exist solutions $\gamma\in S_{v_i}$ and $\gamma'\in S_{v_{i+1}}$ that are adjacent in $\SSS(G)$ (Lemma~\ref{lem:characterizationCSGs}(\ref{pr:adjacency})). Clearly, $\alpha\in S_{v_0}$ and $\beta\in S_{v_k}$. 
Combining these facts yields a path from $\alpha$ to $\beta$ in $\SSS(G)$.
\qed\end{proof}

\noindent
Lemma~\ref{propo:one}
implies
that for a solution graph concept $\SSS$ and {\em any} terminal projection $p$ and terminal set $T$, we can 
decide 
{\sc $\SSS$-Connectivity}
 if we know $\SSS^c(G,T)$ (the answer is YES if and only if $\SSS^c(G,T)$ is connected) and the {\sc $\SSS$-Reachability} problem if we know $\SSS^c(G,T)$ and the $\alpha$-node and the $\beta$-node (the answer is YES if and only if these two nodes are in the same component). 
However,
for obtaining an {\em efficient} algorithm using this strategy, we 
must throw away enough irrelevant information to ensure that $\SSS^c(G,T)$ will be significantly smaller than $\SSS(G)$, yet 
maintain enough information to ensure
 the efficient computation of
 $\SSS^c(G,T)$, 
 without first constructing $\SSS(G)$. 
Our strategy for doing this is to use dynamic programming to
compute $\SSS^c(H,T')$ for ever larger subgraphs $H$ of $G$, while ensuring that all of the \CSG s stay small throughout the process.
The remainder of this paper shows 
a successful example of 
this strategy. 

\section{Dynamic Programming Rules for Recoloring}
\label{sec:DP}

The following terminology is based on widely used techniques for dynamic programming over tree decompositions; see  
Section~\ref{sec:chordalgraphs} 
and~\cite{BBL13,Kl94,Ni06} 
for further information. 

A {\em terminal graph $(G,T)$} is a graph $G$ together with a vertex set $T\subseteq V(G)$, 
whose vertices are called the {\em terminals}. 
If $T=V(G)$, then $(G,T)$ is called a {\em leaf}.
If $v\in T$, then we say that the new terminal graph $(G,T\bs \{v\})$ is obtained from $(G,T)$ by {\em forgetting $v$} (or {\em using a forget operation}).
If $T\not=V(G)$,
$v\in T$ and $N(v)\subseteq T$ then we say that $(G,T)$ can be obtained from $(G-v,T\bs \{v\})$ by {\em introducing $v$} (or {\em using an introduce operation}). 
Note that for a terminal graph $(G',T')$ with $T'\not=\emptyset$, different graphs can be obtained from $(G',T')$ by introducing a vertex $v$, whereas forgetting a terminal always yields a unique result. 
The condition that 
each neighbor of the new vertex~$v$ must be in~$T$ is necessary, as we will see at several places in our proofs.
We say that $(G,T)$ is the {\em join of $(G_1,T)$ and $(G_2,T)$}  if
\begin{itemize}
 \item $G_1$ and $G_2$ are induced subgraphs of $G$,
 \item $V(G_1)\cap V(G_2)=T$ and $V(G_1)\cup V(G_2)=V(G)$,
 \item $V(G_1)\not=T$ and $V(G_2)\not=T$, and
 \item for every $uv\in E(G)$, it holds that $uv\in E(G_1)$ or $uv\in E(G_2)$.
\end{itemize}
We will now focus on CSGs for the $k$-color graph concept $\C_k$, using the terminal projection $p(G,T,\gamma)=\gamma\restr_{T}$. 
We will show how to compute the \CSG\ $\C^c_k(G,T)$ when $(G,T)$ is obtained using a forget, introduce or join operation from a (pair of) graph(s) for which we know the \CSG (s).
Hatanaka, Ito and Zhou~\cite{HIZ15} considered a variant of these \CSG s,
namely 
for the case that $|T|=1$ in the 
context of list colorings of caterpillars. They gave a combined introduce and forget rule and a restricted type of join rule, similar to the above rules.

We first state the (trivial) rule for computing $\C^c_k(G,T)$ for leaves, which follows from the facts that $\C_k(G)$ has $k$-colorings of~$G$ as nodes and that the label~$\ell(x)$ of a node~$x$ in $\C^c_k(G,T)$ is a $k$-coloring of $G[T]$.
Afterwards we give  the rules for the forget, introduce and join operations, which can be proven straightforwardly using
Lemma~\ref{lem:characterizationCSGs}.\footnote{Formal proofs of Lemmas~\ref{lem:recolForget}--\ref{lem:recolJoin} can be found in Appendix~\ref{a-rules}.}
Together these four rules allow us to compute the \CSG\ $\C^c_k(G,T)$ for every terminal graph~$(G,T)$. 
Figure~\ref{fig:unitintCSG} illustrates the rules for the forget and introduce operation, whereas Figure~\ref{fig:bigstar} illustrates the rule for the join operation. 

\begin{lem}[Leaf]
\label{propo:recolLeaf}
Let $(G,T)$ be a terminal graph with $T=V(G)$. 
Then $\C^c_k(G,T)$ is isomorphic to $\C_k(G)$ and its label function~$\ell$ is the isomorphism from $\C^c_k(G,T)$ to $\C_k(G)$. 
Moreover, for every $k$-coloring $\gamma$ of $G$, the $\gamma$-node of $\C^c_k(G,T)$ is the node $v$ with $\ell(v)=\gamma$. 
\end{lem}

\def\Cnew{H'}  
\def\Corig{H}  

\begin{lem}[Forget]
\label{lem:recolForget}
Let $(G,T)$ be a terminal graph.
For every $v\in T$, 
it holds that
$(H',\ell')=\C^c_k(G,T\bs \{v\})$ can be computed from $(H,\ell)=\C^c_k(G,T)$ as follows:
\begin{itemize}
 \item For every node $x$ in $\Corig$ with $\ell(x)=\gamma$, let
 $\ell'(x)=\gamma\restr_{T\bs\{v\}}$.
 \item Iteratively contract every edge between two nodes $x$ and $y$ with $\ell'(x)=\ell'(y)$ and assign
 label $\ell'(z):=\ell'(x)$ to the resulting node $z$. 
 \end{itemize}
Moreover, for any coloring $\gamma$ of $G$, the $\gamma$-node of $\C^c_k(G,T\bs \{v\})$ is the node that results from contracting the set of nodes that includes the $\gamma$-node of $\C^c_k(G,T)$.
\end{lem}

\def\Cnew{H'}	
\def\Corig{H}

\begin{lem}[Introduce]
\label{lem:recolIntroduce}
Let $(G,T)$ be a terminal graph obtained from a terminal graph $(G-v,T\setminus \{v\})$ by introducing $v$.
Then $(\Cnew,\ell')=C^c_k(G,T)$ can be computed as follows from $(\Corig,\ell)=\C^c_k(G-v,T\setminus \{v\})$:
\begin{itemize}
 \item For every node $x$ of $\Corig$ with label~$\ell(x)$, and 
 every color $c\in \{1,\ldots,k\}$: if the (unique) function $\delta:T\to \{1,\ldots,k\}$ with $\delta(v)=c$ and $\delta\restr_{T}=\ell(x)$ is a 
 coloring of $G[T]$ then introduce a node $x_c$ with label $\ell'(x_c)=\delta$.
 \item 
 For every pair of distinct nodes $x_c$ and $y_d$: add an edge between them if and only if
 (1) $x=y$ or (2) $xy$ is an edge in $\Corig$ and $c=d$. 
\end{itemize} 
Moreover, for every $k$-coloring $\gamma$ of $G$, if $x$ is the $\gamma\restr_{V(G)\bs \{v\}}$-node in $\Corig$ and $\gamma(v)=c$, then $x_c$ is the $\gamma$-node of $\Cnew$.
\end{lem}

\def\Cnew{H}	
\def\Cone{H_1}	
\def\Ctwo{H_2}	

\begin{lem}[Join]
\label{lem:recolJoin}
Let $(G,T)$ be a terminal graph that is the join of terminal graphs $(G_1,T)$ and $(G_2,T)$.
Let $(\Cone,\ell_1)=C^c_k(G_1,T)$ and $(\Ctwo,\ell_2)=C^c_k(G_2,T)$.
Then $(\Cnew,\ell)=C^c_k(G,T)$ can be computed as follows:
\begin{itemize}
\item For every pair of nodes $x\in V(\Cone)$
 and $y\in V(\Ctwo)$: if $\ell_1(x)=\ell_2(y)$ then introduce a node $(x,y)$ with $\ell((x,y))=\ell_1(x)$. 
\item 
 For two distinct nodes $(x,y)$ and $(x',y')$, add an edge between them if and only if $xx'$ is an edge in $\Cone$ and $yy'$ is an edge in $\Ctwo$.
\end{itemize}
Moreover, for every $k$-coloring $\gamma$ of $G$, if $x$ is the $\gamma\restr_{V(G_1)}$-node in
$\Cone$ and $y$ is the $\gamma\restr_{V(G_2)}$-node in
$\Ctwo$, then $(x,y)$ is the $\gamma$-node in $\Cnew$.
\end{lem}

\begin{figure}
\centering
\scalebox{0.5}{$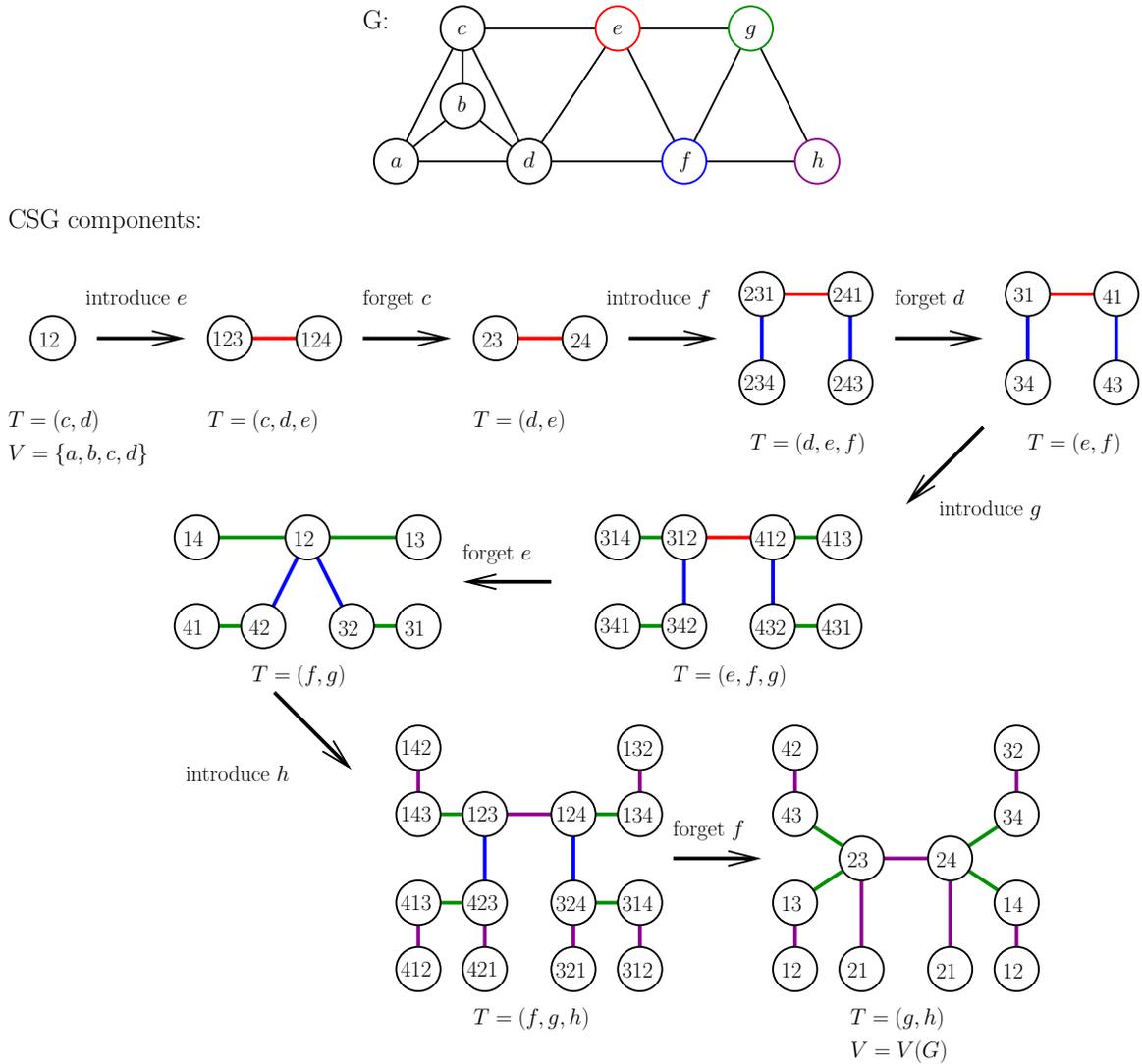$}
\caption{An example of computing CSGs using forget and introduce operations. A 4-colorable 2-connected chordal 
graph $G$ with $V(G)=\{a,b,c,d,e,f,g,h\}$ is shown. 
Note that $G$ is in fact unit interval and isomorphic to 
the graph $G^I_8$ defined in 
Section~\ref{sec:DP}.
Starting with one component of the CSG $\C^c_4(G[\{a,b,c,d\}],\{c,d\})$,
the corresponding component of $\C^c_4(G,\{g,h\})$ is computed, using four forget and introduce operations. The $G[T]$-colorings in the node labels are given as sequences of colors for the ordered version of $T$
as indicated below each CSG.
For instance, for $T=(c,d)$, the node label $12$ indicates the coloring $\gamma$ with $\gamma(c)=1$ and $\gamma(d)=2$.}
\label{fig:unitintCSG}
\end{figure}

\noindent
{\bf Remark 1.} 
The DP rules in this section can be generalized further to capture the rules of~\cite{HIZ15} for
the {\em list coloring} generalization $\C_L$ of $\C_k$. In this generalization, an instance $G,L$ consists of a graph $G$ together with color lists $L(v)\subseteq \{1,\ldots,k\}$ for each $v\in V(G)$. Solutions are now list colorings, which are colorings $\alpha$ of $G$ such that $\alpha(v)\in L(v)$ for each $v\in V(G)$. Adjacency is defined as before. So the {\em list coloring solution graph} $\C_L(G,L)$ is an induced subgraph of $\C_k(G)$.
Hence, it is straightforward to generalize our DP rules to $\C_L$, namely by simply omitting all nodes that correspond to invalid vertex colors.

\section{Examples of Exponential Size CSGs}
\label{sec:badexamples}

In this section we further
illustrate the dynamic programming rules given in Section~\ref{sec:DP}
and show that  components of $\C^c_k(G)$ can grow exponentially, even if $G$ is a chordal graph and $k=4$.
Namely, when considering 4-colorable chordal graphs that may have cut vertices, it is easy to obtain CSGs that have exponentially large components: take $p$ copies of the graph shown in Figure~\ref{fig:cutvertCSG}(a), and identify the $g$-vertices of all of these graphs. Call the resulting graph $G^*_p$. The graph $G^*_2$ is shown in Figure~\ref{fig:bigstar}(a).

\begin{propo}
\label{propo:ex1}
For every integer $p\ge 1$, 
$\C^c_4(G^*_p,\{g\})$ has a component with $1+3\cdot 2^p$ nodes.
\end{propo}
\begin{proof}
By induction over $p$ we prove the following: $\C^c_4(G^*_p,\{g\})$ has a component that is a star (a $K_{1,n}$ graph) in which the central node has label 1 (to be precise, this means that the label is a coloring that assigns color 1 to vertex $g$), and which has $2^p$ leaves with label $j$, for $j\in \{2,3,4\}$. 
The case $p=1$ can easily be verified; see also Figure~\ref{fig:cutvertCSG}. 
For the induction step, apply Lemma~\ref{lem:recolJoin} to the star components of $\C^c_4(G^*_{p-1},\{g\})$ and $\C^c_4(G^*_1,\{g\})$ given by the induction 
hypothesis:
for $j\in \{2,3,4\}$, the $2^{p-1}$ nodes with label~$j$ of the former graph are combined with two nodes with label~$j$ of the latter graph, giving $2^p$ new nodes with label~$j$. All of these are adjacent only to the unique new node with label 1.\qed 
\end{proof}

\begin{figure}
\centering
\scalebox{0.5}{$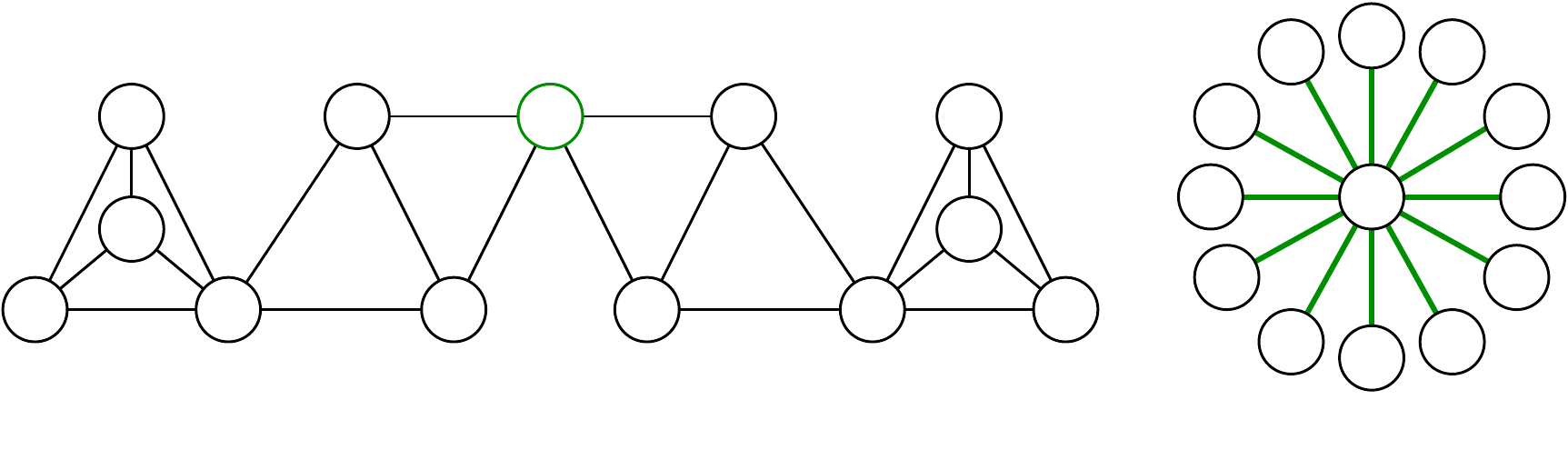$}
\caption{(a) 
The graph $G^*_2$, which can be obtained using a join operation on two subgraphs isomorphic to the graph $G$ shown in Figure~\ref{fig:cutvertCSG}(a), with $T=\{g\}$. (b) One component of the CSG $\C^c_4(G^*_2,\{g\})$, which can be obtained by combining two copies of the CSG from Figure~\ref{fig:cutvertCSG}(d), as shown in Lemma~\ref{lem:recolJoin}.
}
\label{fig:bigstar}
\end{figure}

With a little more effort, we can also construct CSGs with exponentially large components when restricting to $(k-2)$-connected $k$-colorable chordal graphs, or even 2-connected 4-colorable unit interval graphs, as follows. 
For $p\ge 4$, let the graph $G^I_p$ have vertex set $\{v_0,\ldots,v_{p-1}\}$, and edge set $\{v_0v_3\} \cup \{v_iv_{i+1} \mid 0\le i\le p-2\} \cup \{v_iv_{i+2} \mid 0\le i\le p-3\}$. A graph isomorphic to $G^I_8$ is shown in Figure~\ref{fig:unitintCSG}.  
Note that each $G^I_p$ is unit interval.
For our proof 
we need the following simple observation (which has been used in various earlier papers on recoloring, such as~\cite{BC09}).

\begin{lem}
\label{propo:degeneracy}
Let $\alpha$ and $\beta$ be two $k$-colorings of a graph $G=(V,E)$, and let $v$ be a vertex of degree at most $k-2$. Then $\C_k(G)$ contains a path from $\alpha$ to $\beta$ if and only if $\C_k(G-v)$ contains a path from $\alpha\restr_{V\bs \{v\}}$ to  $\beta\restr_{V\bs \{v\}}$.
\end{lem}
Below we state our claim more precisely and give a proof of it as well.

\begin{propo}
\label{propo:ex2}
For $p=4q+4$ with $q\in \mathbb{N}$, the CSG $\C^c_4(G^I_p,\{v_{p-2},v_{p-1}\})$ has $4!$ components on at least $2^q$ nodes.
\end{propo}
\begin{proof}
For every set $S\subseteq \{1,\ldots,q\}$, we construct a coloring $\alpha_S$ of $G^I_p$ as follows. For all $j\in \{0,\ldots,q\}$:
\begin{itemize}
 \item $\alpha_S(v_{4j})=3$ if $j\in S$, and $\alpha_S(v_{4j})=4$ if $j\not\in S$.
 \item $\alpha_S(v_{4j+1})=4$ if $j\in S$, and $\alpha_S(v_{4j+1})=3$ if $j\not\in S$.
 \item $\alpha_S(v_{4j+2})=1$.
 \item $\alpha_S(v_{4j+3})=2$.
\end{itemize}
Observe that for every $S$, $\alpha_S$ is a 4-coloring of $G^I_p$. 
There are $2^q$ possible choices of $S$, and therefore $2^q$ such colorings $\alpha_S$. 
An induction proof based on Lemma~\ref{propo:degeneracy} shows that for every $S_1\subseteq \{1,\ldots,q\}$ and $S_2\subseteq \{1,\ldots,q\}$, $\C_4(G)$ contains a path from $\alpha_{S_1}$ to $\alpha_{S_2}$: informally, vertex $v_{p-1}$ has degree $2$ and is therefore irrelevant for the reachability question. After deleting $v_{p-1}$, $v_{p-2}$ has degree 2, and may be deleted next. Continuing this procedure ends with two 4-colorings of the complete graph on vertices $\{v_0,v_1,v_2,v_3\}$, which coincide.
That is, for every~$S$, $\alpha_S$ assigns the colors $4,3,1,2$ to the vertices $v_0,v_1,v_2,v_3$, respectively
(note that no vertex of  $\{v_0,v_1,v_2,v_3\}$ can be recolored).
It follows that all of the colorings $\alpha_S$ we constructed are part of the same component of $\C_4(G^I_p)$. Finally, we observe that every coloring $\alpha_S$ forms a one-node label component in $\C^c_4(G^I_p,\{v_{p-2},v_{p-1}\})$. Indeed, the only vertex that can be recolored in any $\alpha_S$ is the vertex $v_{p-1}$; all other vertices have three distinctly-colored neighbors.
Summarizing, $\C^c_4(G^I_p,\{v_{p-2},v_{p-1}\})$ contains a component that contains at least $2^q$ nodes  that are labeled with a coloring that assigns colors 1 and 2 to vertices $v_{p-2}$ and $v_{p-1}$, respectively. 

For every 4-coloring of $G[\{v_0,v_1,v_2,v_3\}]$, the CSG contains a component isomorphic to the component considered above, so there are $4!$ components of this type.\qed
\end{proof}

The last CSG shown in Figure~\ref{fig:unitintCSG} contains two nodes with label $12$; these correspond to the colorings $\alpha_{\emptyset}$ and $\alpha_{\{1\}}$ constructed in the above proof. The CSG shows that any recoloring sequence between these two colorings needs to recolor the vertices $v_{p-2}=g$ and $v_{p-1}=h$ at least 
two resp. three times.
We remark that the proofs of Propositions~\ref{propo:ex1} 
and~\ref{propo:ex2}
illustrate different proof techniques for CSGs: one uses the dynamic programming rules, and the other argues about label components of the solution graph directly.
Both examples show that we need to do more than only computing CSGs to solve the problem for $(k-2)$-connected chordal graphs in polynomial time. 
In the next section we will characterize the CSGs and show that it suffices to compute only a part of them.

\section{Recoloring Chordal Graphs}
\label{sec:chordalgraphs}

In this section we will show that \CSG s can be used to efficiently decide the {\sc $\C_k$-Reachability} problem for $(k-2)$-connected chordal graphs. Recall that a graph is chordal if it contains no induced cycle of length greater than~3. 
To prove the result, we use the fact that for a chordal graph $G$ and any clique~$T$ of~$G$, the terminal graph $(G,T)$  can recursively be constructed from simple cliques using a polynomial number of clique-based introduce, forget and join operations. 
In order to do this we first define the notion of a nice tree decomposition for terminal graphs in  Section~\ref{ssec:nice1}.
In the same section we show an upper bound on the size of an arbitrary nice tree decomposition for a terminal graph.
Afterwards we make the above fact for chordal graphs precise in Section~\ref{ssec:Nicetds}, namely by defining chordal nice tree decompositions for terminal graphs.
We remark that some of our statements are similar to (and can alternatively be deduced from) well-known facts about tree decompositions~\cite{Di10} and nice tree decompositions~\cite{Kl94} for graphs.
However, for readability, and since we need to prove an upper bound on the size of a nice tree decomposition for a terminal graphs, we give a self-contained presentation.

\subsection{Nice Tree Decompositions for Terminal Graphs}
\label{ssec:nice1}

Nice tree decompositions
describe how a terminal graph $(G,T)$ can be obtained from trivial graphs using forget, introduce and join operations.
A {\em nice tree decomposition} of a terminal graph $(G,T)$
(where $G$ is an arbitrary graph, not necessarily chordal, and $T$ may not be a clique) 
is a tuple $(\TTT,X,r)$, 
where
$\TTT$ is a tree with root $r$ and $X$ is an assignment of {\em bags} $X_u\subseteq V(G)$ for each $u\in V(\TTT)$ that can be defined recursively as follows:
\begin{itemize}
 \item [(1)] If $T=V(G)$, then the tree $\TTT$ consists of one (root) node $r$ with bag $X_r=T$.
 \item [(2)] If  $v\in V(G)\bs T$
 and $(\TTT',X,r')$ is a nice tree decomposition of $(G,T\cup \{v\})$, then a nice tree decomposition for $(G,T)$ can be obtained by adding a new root $r$ with $X_r=T$ and adding the edge $rr'$.
 \item [(3)] If $(G,T)$ can be obtained from $(G-v,T\bs \{v\})$ using an introduce operation and $(\TTT',X,r')$ is a nice tree decomposition of $(G-v,T\bs \{v\})$, then a nice tree decomposition for $(G,T)$ can be obtained by adding a new root $r$ with $X_r=T$ and adding the edge $rr'$.
 \item [(4)] If $(G,T)$ can be obtained from $(G_1,T)$ and $(G_2,T)$ using a join operation, and $(\TTT_1,X,r_1)$ and $(\TTT_2,X,r_2)$ are nice tree decompositions of $(G_1,T)$ and $(G_2,T)$, then a nice tree decomposition for $(G,T)$ can be obtained by adding a new root $r$ with $X_r=T$ and adding edges $rr_1$ and $rr_2$. 
\end{itemize}
We call a node $u\in V(\TTT)$ a {\em leaf}, {\em forget node}, {\em introduce node} or {\em join node} if 
$u$
is added as the root in case (1), (2), (3) or (4), respectively. 
The {\em width} of $(\TTT,X,r)$ is $\max_{u\in V(\TTT)} |X_u|-1$. 

We will need the following lemma.

\begin{lem}
\label{lem:SizeNiceTreedecomp}
Let $(\TTT,X,r)$ be a nice tree decomposition of a terminal graph~$(G,T)$ of width at most 
$w\geq 1$, and let $n=|V(G)|\geq 1$.
Then $|V(\TTT)|\le (w+4)n$.
\end{lem}
\begin{proof}
Let $t=|T|$.
We use induction over $|V(\TTT)|$ to prove that $$|V(\TTT)|\le 2n-t+(w+2)(\max\{0,n-t-1\}).$$
Then, since this value is at most $(w+4)n$, the lemma statement follows. 

\medskip
\noindent
Let $|V(\TTT)|=1$. 
Then
the root $r$ of $\TTT$ is a leaf (so $T=V(G)$ and $t=n$).
Hence, we have that
$$|V(\TTT)|=1\le n=2n-t+(w+2)(\max\{0,n-t-1\}).$$
Let $|V(\TTT)|\geq 1$. Then the root~$r$ is either a forget, introduce or join node. We consider each of these cases below.\\
If $r$ is a forget node then by induction, after adding 
the new root to the tree,
the number of nodes is at most:
\[
1 + 2n - (t+1) + (w+2)(\max\{0,n-(t+1)-1\})\ \le\ 
2n-t + (w+2)(\max\{0,n-t-1\}).
\]
If $r$ is an introduce node then by induction, after adding
the new root to the tree,
the number of nodes is at most:
\[
1 + 2(n-1)-(t-1) + (w+2)(\max\{0,n-1-(t-1)-1\}) = 
\]
\[
2n-t + (w+2)(\max\{0,n-t-1\}).
\]
Finally, suppose that $r$ is a join node and that $(G,T)$ is obtained by joining together graphs on $n_1$ and $n_2=n-n_1+t$ nodes. 
From the definition of the join operation it follows that both of these values are strictly larger than $t$, so we may write $\max\{0,n_1-t-1\}=n_1-t-1$ and  $\max\{0,n_2-t-1\}=n_2-t-1=n-n_1-1$.
Then by induction, after adding 
the new root,
the number of nodes is at most:
\[
1\  +\  2n_1-t + (w+2)(n_1-t-1)\ +\ 2(n-n_1+t) -t + (w+2)(n-n_1-1)\ =
\]
\[
1 + 2n + (w+2)(n-t-2)\ 
\le\ 
2n - t + (w+2)(n-t-1).
\]
For the last step, we used that $t\le (w+1)$. 
\qed\end{proof}
 
\subsection{Chordal Nice Tree Decompositions for Terminal Graphs}
\label{ssec:Nicetds} 

A nice tree decomposition $(\TTT,X,r)$ of $(G,T)$ is {\em chordal} if for every node $u\in V(\TTT)$, $X_u$ is a clique 
of~$G$. 
Note that,
if $(\TTT,X,r)$ is a chordal nice tree decomposition of a $k$-colorable graph~$G$, then 
the width of $(\TTT,X,r)$ is at most $k-1$. 
Hence, Lemma~\ref{lem:SizeNiceTreedecomp} shows that any chordal nice tree decomposition of a $k$-colorable graph has at most $(k+3)n$ nodes.

In order to show how to find a chordal nice tree decomposition in polynomial time we need the following lemma, which tells us how to select the proper type of root node when constructing such a tree decomposition.
Here, a terminal graph $(G_1,T_1)$ is called {\em smaller} than another terminal graph $(G_2,T_2)$ if 
$2|V(G_1)|-|T_1| < 2|V(G_2)|-|T_2|$.

\begin{lem}
\label{lem:ChordalNiceTreedecomp}
Let $(G,T)$ be a terminal graph where $G=(V,E)$ is a chordal graph, and $T$ is a clique with $T\not=V$.
If $G-T$ is disconnected, then $(G,T)$ can be obtained from a pair of smaller chordal terminal graphs $(G_1,T)$ and $(G_2,T)$ using a join operation.  Otherwise, $(G,T)$ can be obtained from a smaller chordal terminal graph $(G',T')$ using either a forget or introduce operation, where $T'$ is a clique. 
For every such $(G,T)$, the relevant operation and subgraph(s) can be found in polynomial time.
\end{lem}

\begin{proof}
If $G-T$ is disconnected, then let $C$ be the vertex set of a component of $G-T$, and consider the two graphs $G_1=G[T\cup C]$ and $G_2=G[V\bs C]$. Then $(G,T)$ is the join of $(G_1,T)$ and $(G_2,T)$ with $|V(G_1)|-|T|<|V|-|T|$ and $|V(G_2)|-|T|<|V|-|T|$. 

Next assume $G-T$ is connected. If $T$ contains a vertex $v$ that has no neighbors in $G-T$, then $(G,T)$ can be obtained from $(G-v,T\bs \{v\})$ using an introduce operation, and 
$2|V(G-v)|-|T\bs\{v\}|=2|V|-|T|-1$.

Finally, assume that $G-T$ is connected and every vertex in $T$ is adjacent to at least one vertex in $G-T$. Then we prove that there exists a vertex $u\in V\bs T$ that is adjacent to every vertex in $T$. Let $u\in V\bs T$ be a vertex with a maximum number of neighbors in $T$. 
Suppose for a contradiction that at least one vertex in $T$ is not adjacent to $u$. 
Then we can choose  
a {\em shortest} path $P$ in $G-T$ from $u$ to a vertex $v$ with $T\cap (N(v)\bs N(u))\not=\emptyset$. (Such a $v$ and $P$ exist because every vertex in $T$ has a neighbor outside of $T$ and $G-T$ is connected.) Let $w$ be the {\em last} vertex on $P$ (when going from $u$ to $v$) with $T\cap (N(w)\bs N(v))\not=\emptyset$. Since $u$ satisfies this condition (because it has a maximum number of neighbors in $T$), such a vertex $w$ exists. 
Now we have chosen distinct vertices $w$ and $v$ such that there exists a path $P'$ between them (namely the sub path of $P$ from $w$ to $v$) with the following property: for all internal vertices $x$ of $P'$, $T\cap N(x)\subseteq T\cap N(u)\cap N(v)$ holds. 
Choose $y\in T\cap (N(v)\bs N(u)$ and $z\in T\cap (N(w)\bs N(v))$ (by the choice of $v$ and $w$, such vertices exist). 
Now combining the path $P'$ from $w$ to $v$ with the edges $vy$, $yz$ and $wz$ yields an induced cycle in $G$ of length at least~4, contradicting that $G$ is chordal. 

We conclude that there exists a vertex $u\in V\bs T$ with $T\subseteq N(u)$. So $(G,T)$ can be obtained from $(G,T\cup \{u\})$ by a forget operation, such that $T\cup \{u\}$ is a clique in $G$, and $2|V|-|T\cup \{u\}|=2|V|-|T|-1$.

The above case study can easily be translated to a polynomial-time algorithm for finding the graph operation that applies.
\qed
\end{proof}

We are now ready to state the following result.

\begin{corol}
\label{corol:chordalNiceTreedecomp}
Let $G$ be a chordal $k$-colorable graph on $n$ vertices, and let $T$ be a clique of $G$. Then it is possible to find a chordal nice tree decomposition of $(G,T)$ on at most $(k+3)n$ nodes in polynomial time.
\end{corol}
\begin{proof}
Lemma~\ref{lem:ChordalNiceTreedecomp} shows how we can choose the proper type of root node. We can build the
chordal
 nice tree decomposition by adding this node to the tree decomposition(s) of (a) smaller graph(s). 
The entire chordal nice tree decomposition is constructed by continuing this process recursively. Lemma~\ref{lem:SizeNiceTreedecomp} shows that the resulting chordal nice tree decomposition has at most $(w+4)n$ nodes, where $w+1$ is the maximum bag size. Since every bag is a clique of $G$ and the graph is $k$-colorable, we have $w+1\le k$, so there are at most $(k+3)n$ nodes. 
Since we have a polynomial number of nodes, and for every node we spend polynomial time (Lemma~\ref{lem:ChordalNiceTreedecomp}), the entire process terminates in polynomial time. 
\qed\end{proof}

We note that the precise complexity bound 
in Corollary~\ref{corol:chordalNiceTreedecomp}
depends on implementation details, which are beyond the scope of this paper.

\subsection{The Structure of CSGs for $(k-2)$-Connected Chordal Graphs}\label{ssec:invariant}

Using an inductive proof based on Lemma~\ref{lem:ChordalNiceTreedecomp}, we will now characterize the shape of CSGs for $(k-2)$-connected $k$-colorable chordal graphs.
We start with a lemma that we will apply to $(k-2)$-connected $k$-colorable chordal graphs in our induction proofs.
\begin{lem}
\label{propo:connectedness}
Let $G$ be a $\ell$-connected chordal graph, and let $T$ be a clique of $G$ with $T\not=V(G)$. 
If $(G,T)$ can be obtained from $(G-v,T\bs \{v\})$ using an introduce operation, then $|T|\ge \ell+1$ and $G-v$ is $\ell$-connected.
If $(G,T)$ can be obtained from $(G_1,T)$ and $(G_2,T)$ using a join operation, then $|T|\ge \ell$ and both $G_1$ and $G_2$ are $\ell$-connected.
\end{lem}
\begin{proof}
If $(G,T)$ is obtained from $(G-v,T\bs \{v\})$ using an introduce operation, then $N(v)\subseteq T\bs \{v\}$ 
by definition. 
Since $T\not=V(G)$, it follows that $T\bs \{v\}$ is a vertex cut of $G$ that separates $v$ from at least one other vertex, so $|T|=|T\bs \{v\}|+1\ge \ell+1$. In addition, since $T$ is a clique of $G$, every vertex cut of $G-v$ is also a vertex cut of $G$, and therefore $G-v$ is $\ell$-connected.

If $(G,T)$ is obtained from $(G_1,T)$ and $(G_2,T)$ using a join operation, then 
$T$ is a vertex cut of $G$ that separates $V(G_1)\bs T$ from $V(G_2)\bs T$, so $|T|\ge \ell$. 
In addition, since $T$ is a clique of $G$, every vertex cut of $G_i$ is a vertex cut of $G$ (for $i=1,2$),
and therefore $G_1$ and $G_2$ are $\ell$-connected.
\qed\end{proof}

We need some extra definitions.
For integers $m,k$ with $1\le m\le k$, a labeled graph $(H,\ell)$ is an {\em $(m,k)$-color-complete graph} if there exists a set $T$ with $|T|=m$ such that:
\begin{itemize}
 \item for all vertices $v\in V(H)$, $\ell(v)$ is a $k$-coloring of a complete graph on vertex set~$T$,  
 \item every such $k$-coloring of $T$ appears at exactly one vertex of $H$, and
 \item two vertices of $H$ are adjacent if and only if their labels differ on exactly one element of~$T$. 
\end{itemize}
From this definition it follows that for every pair of integers $m$ and $k$, there is a unique $(m,k)$-color complete graph, up to the choice of 
$T$. 
An $(m,k)$-color-complete graph has $k!/(k-m)!$ vertices (this is the number of ways to $k$-color a complete graph on $m$ vertices), and every vertex has degree $m(k-m)$. In particular, if $m=k$ then the graph consists of $k!$ isolated vertices (meaning that the graph is a forest). 
A labeled graph $(H,\ell)$ is said to satisfy the {\em injective neighborhood property} (INP) if for every vertex $u\in V(H)$ and every pair of distinct neighbors $v,w\in N(u)$, it holds that $\ell(v)\not=\ell(w)$. 
Note that $(m,k)$-color-complete graphs trivially satisfy the INP.

We will now show that for the graphs we consider, the following invariant is maintained by introduce, forget and join operations: 

\medskip
\noindent
{\it The CSG is an $(m,k)$-color complete graph, or a forest that satisfies the INP.}

\medskip
\noindent
Note that a $(k,k)$-color complete graph is trivially a forest that satisfies the INP.
We start with the trivial observation that this invariant initially holds.

\begin{lem}
\label{propo:invariantLeaf}
Let $G=(V,E)$ be a complete graph on $m$ vertices, with $m\le k$. 
Then $\C^c_k(G,V)$ is an $(m,k)$-color complete graph.
\end{lem}

We now prove that a {\em forget} operation maintains the invariant (below, we argue that all the relevant cases are covered by the next lemma).
Recall that a label $\ell(u)$ of a node $u$ of $\C^c_k(G,T)$ is a coloring of $G[T]$, so by $\ell(u)(x)$ we denote the color that $x\in T$ receives in this coloring.

\begin{lem}
\label{lem:invariantForget}
Let $G$ be a 
$k$-colorable chordal graph and let $T$ be a clique of $G$ with $k-1\le |T|$, and $v\in T$. 
If $\C^c_k(G,T)$ is a $(k-1,k)$-color complete graph, then $\C^c_k(G,T\bs \{v\})$ is a $(k-2,k)$-color complete graph. 
If $\C^c_k(G,T)$ is a forest that satisfies the INP, then $\C^c_k(G,T\bs \{v\})$ is a forest that satisfies the INP. 
\end{lem}

\begin{proof}
Let $(H,\ell)=\C^c_k(G,T)$ and $(H',\ell')=\C^c_k(G,T\bs \{v\})$.
We will use that $(H',\ell')$ can be constructed from $(H,\ell)$ as shown in Lemma~\ref{lem:recolForget}.

First consider the case that $(H,\ell)$ is a $(k-1,k)$-color-complete graph.
Then, from the definition of  $(k-1,k)$-color-complete graphs it follows that for every coloring $\alpha$ of $G[T\bs \{v\}]$, the nodes $\{x\in V(H) \mid \ell(x)\restr_{T\bs \{v\}}=\alpha\}$ induce a nonempty complete subgraph of $H$. 
When constructing $(H',\ell')$ from $(H,\ell)$, 
this subgraph will be contracted into one node, so for every such coloring $\alpha$, $H'$ contains exactly one node with label $\alpha$. 
Consider two colorings $\alpha_1$ and $\alpha_2$ of $G[T\bs \{v\}]$ that differ on only one vertex $w\in T\bs \{v\}$. We can extend both to a coloring of $G[T]$ by choosing a color for $v$ that occurs in neither $\alpha_1$ nor $\alpha_2$ (since $|T\bs \{v\}|=k-2$), which yields colorings of $G[T]$ that are adjacent in $H$ (since $(H,\ell)$ is $(k-1,k)$-color complete) and that are compatible with $\alpha_1$ resp.\ $\alpha_2$.
It follows that the nodes of $H'$ with labels $\alpha_1$ and $\alpha_2$ are adjacent (Lemma~\ref{lem:recolForget}). We conclude that $(H',\ell')$ is a $(k-2,k)$-color-complete graph. 

Next, consider the case that $(H,\ell)$ is a forest that satisfies the INP. Then $H'$ is clearly a forest, since it can be obtained by contracting $H$ (Lemma~\ref{lem:recolForget}). If $H$ contains no edges, then $H'$ contains no edges, and trivially 
$(H',\ell')$ satisfies the INP. Note that this is the case if $|T|=k$.
So it only remains to consider the case that $H$ contains at least one edge, and therefore $|T|=k-1$. 
The last part of the proof is illustrated in Figure~\ref{fig:invariantForget}.

\begin{figure}
\centering
\scalebox{0.6}{$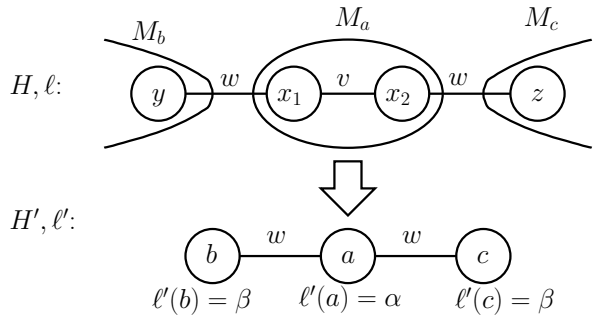$}
\caption{An illustration of the proof of Lemma~\ref{lem:invariantForget}.}
\label{fig:invariantForget}
\end{figure}

For every node $a\in V(H')$, denote by $M_a$ the set of nodes of $H$ that are contracted to obtain $a$, when constructing $(H',\ell')$ from $(H,\ell)$ as shown in Lemma~\ref{lem:recolForget} (so all nodes in $M_a$ are labeled with a $G[T]$-coloring that is compatible with the $G[T\bs \{v\}]$-coloring $\ell'(a)$). 
For every $k$-coloring $\alpha$ of $G[T\bs \{v\}]$, there are at most two compatible colorings of $G[T]$, since $|T\bs \{v\}|=k-2$. So since $H$ satisfies the INP, the subgraph of $H$ induced by the nodes that have an $\alpha$-compatible label has maximum degree at 
most~1, and thus maximum component size at 
most~2. 
It follows that for every $a\in V(H')$, $|M_a|\le 2$. 

We now prove that $(H',\ell')$ satisfies the INP. 
Suppose to the contrary that $H'$ contains a node $a$ with $\ell'(a)=\alpha$, that has two neighbor nodes $b$ and $c$ with $\ell'(b)=\ell'(c)=\beta$. 
Let $w\in T\bs \{v\}$ be the vertex on which $\alpha$ and $\beta$ differ. 
So there are nodes $y\in M_b$ and $x_1\in M_a$ that are adjacent in $H$, and nodes $z\in M_c$ and $x_2\in M_a$ that are adjacent in $H$ (Lemma~\ref{lem:recolForget}). 
Because the adjacent colorings~$\ell(y)$ and~$\ell(x_1)$ differ on vertex $w$, they differ on no other vertex. The same holds for $\ell(z)$ and $\ell(x_2)$. As $H$ satisfies the INP, it follows that $x_1\not=x_2$, so $|M_a|\ge 2$, and thus $|M_a|=2$. 

We conclude that $y,x_1,x_2,z$ is a path in $H$ such that 
$\ell(y)(w)\not=\ell(x_1)(w)$, $\ell(x_1)(v)\not=\ell(x_2)(v)$, and $\ell(x_2)(w)\not=\ell(z)(w)$. 
Recall that labels of adjacent nodes in $H$ differ on exactly one vertex.
The colorings $\ell(y)$, $\ell(x_1)$, $\ell(x_2)$ and $\ell(z)$ all use $|T|=k-1$ different colors out of a total of $k$ colors. Combining these facts shows that $\ell(y)(w)=\ell(x_2)(v)=\ell(z)(v)$.
But since $\ell(y)$ and $\ell(z)$ are both compatible with $\beta$, $\ell(z)(w)=\ell(y)(w)$. This contradicts that $\ell(z)$ is a (proper) coloring of $G[T]$. 
We conclude that $(H',\ell')$ satisfies the INP.
\qed
\end{proof}

Next, we show that the introduce operation maintains the invariant.

\begin{lem}
\label{lem:invariantIntroduce}
Let $G=(V,E)$ be a $(k-2)$-connected $k$-colorable chordal graph and let $T$ be a clique of $G$, with $T\not=V$, such that $(G,T)$ can be obtained from $(G-v,T\bs\{v\})$ using an introduce operation. 
If $\C^c_k(G-v,T\bs \{v\})$ is a $(k-2,k)$-color complete graph, then $\C^c_k(G,T)$ is a $(k-1,k)$-color complete graph. 
If $|T|=k$ or $\C^c_k(G-v,T\bs \{v\})$ is a forest that satisfies the INP, then $\C^c_k(G,T)$ is a forest that satisfies the INP. 
\end{lem}

\begin{proof}
Let $(H,\ell)=\C^c_k(G-v,T\bs \{v\})$ and $(H',\ell)=\C^c_k(G,T)$.
We will use that $(H',\ell')$ can be constructed from $(H,\ell)$ as shown in Lemma~\ref{lem:recolIntroduce}. 
By Lemma~\ref{propo:connectedness}, $|T|\ge k-1$. 
If $|T|=k$, then $H'$ is a set of isolated vertices, which proves the statement. So we may now assume that $|T|=k-1$.  

First consider the case that $(H,\ell)$ is a $(k-2,k)$-color-complete graph.
For every $k$-coloring~$\alpha$ of $G[T]$, there exists exactly one node in $H$ that has a label~$\beta$ that is compatible 
with~$\alpha$. So $H'$ contains exactly one node with label~$\alpha$. 
Consider two colorings~$\alpha_1$ and~$\alpha_2$ of $G[T]$ that differ on exactly one vertex. 
If this vertex is $v$, then the nodes of $H'$ with labels~$\alpha_1$ and~$\alpha_2$ are adjacent (Lemma~\ref{lem:recolIntroduce}). Otherwise, let $\beta_i=\alpha_i\restr_{T\bs \{v\}}$ for $i=1,2$. 
The nodes with labels $\beta_1$ and $\beta_2$ are adjacent in $H$ since it is a color-complete graph. Therefore, the nodes of $H'$ with labels $\alpha_1$ and $\alpha_2$ are also adjacent in this case (Lemma~\ref{lem:recolIntroduce}).
We conclude that $(H',\ell')$ is a $(k-1,k)$-color complete graph.

Next, consider the case that $(H,\ell)$ is a forest that satisfies the INP. From Lemma~\ref{lem:recolIntroduce} it follows easily that $(H',\ell')$ satisfies the INP. 
We now prove that $H'$ is a forest. Since $|T\bs \{v\}|=k-2$, every node $x$ of $H$ has as label $\ell(x)$ a $(k-2)$-coloring of the complete graph $G[T\bs \{v\}]$. So there are exactly two nodes in $H'$ that correspond to $x$, which are adjacent (Lemma~\ref{lem:recolIntroduce}). 

We will now show that for any edge $xy\in E(H)$, the following holds: if $x_1$ and $x_2$ are the vertices of $H'$ that correspond to $x$, and $y_1$ and $y_2$ are the vertices of $H'$ that correspond to $y$, then there is at most one edge in $H'$ with one end in $\{x_1,x_2\}$ and one end in $\{y_1,y_2\}$.
Observe that this property, combined with the fact that $H$ contains no cycles, shows that $H'$ contains no cycles. 

Assume without loss of generality 
that $x_1$ and $y_1$ are adjacent in $H'$. 
Let $w\in T$ be the unique vertex with $\ell'(x_1)(w)\not=\ell'(y_1)(w)$ (note that $w\neq v$).
Observe that the colorings $\ell'(x_1)$ and $\ell'(x_2)$ differ only on $v$, and that the same holds for the colorings $\ell'(y_1)$ and $\ell'(y_2)$. 
Since all colorings in $\ell'$ use $k-1$ colors out of $k$ total colors, it follows that 
$\ell'(x_2)(v)=\ell'(y_1)(w)=\ell'(y_2)(w)$, and $\ell'(y_2)(v)=\ell'(x_1)(w)=\ell'(x_2)(w)$.
Because all of these labels are (proper) colorings, we conclude that $\ell'(x_2)$ differs from the colorings $\ell'(y_1)$ and $\ell'(y_2)$ on both $v$ and $w$, and 
$\ell'(y_2)$ differs from the colorings $\ell'(x_1)$ and $\ell'(x_2)$ on both $v$ and $w$. 
Therefore, $x_1y_1$ is indeed the only edge between these two vertex groups. 
It follows that $H'$ contains no cycles and is a forest.
\qed
\end{proof}

Finally, we show that the join operation maintains the invariant.

\begin{lem}
\label{lem:invariantJoin}
Let $G=(V,E)$ be a
$k$-colorable chordal graph and let $T$ be a clique of $G$, such that $(G,T)$ can be obtained from $(G_1,T)$ and $(G_2,T)$ using a join operation.
If one of $\C^c_k(G_1,T)$ or $\C^c_k(G_2,T)$ is an $(m,k)$-color complete graph, then $\C^c_k(G,T)$ 
equals
the other graph.
If both $\C^c_k(G_1,T)$ and $\C^c_k(G_2,T)$ are forests satisfying the INP, then $\C^c_k(G,T)$ is a forest satisfying the INP.
\end{lem}

\begin{proof}
Let $(H_1,\ell_1)=\C^c_k(G_1,T)$, $(H_2,\ell_2)=\C^c_k(G_2,T)$, and $(H,\ell)=\C^c_k(G,T)$. We use that $(H,\ell)$ can be constructed from $(H_1,\ell_1)$ and $(H_2,\ell_2)$ as shown in Lemma~\ref{lem:recolJoin}. 

First suppose that $(H_1,\ell_1)$ is a color-complete graph. 
Then Lemma~\ref{lem:recolJoin} shows that every node of $H_2$ is combined with exactly one node of $H_1$ (there is exactly one node with the same label), so the nodes of $H$ correspond bijectively to nodes of $H_2$. Furthermore, any edge of $H_2$ is maintained, since $H_1$ has edges between every pair of nodes labeled by colorings that differ on exactly one vertex. So $(H,\ell)$ 
equals
$(H_2,\ell_2)$. If $(H_2,\ell_2)$ is a color-complete graph, the proof is analog. 

So it only remains to prove the statement in the case that both $(H_1,\ell_1)$ and $(H_2,\ell_2)$ are forests that satisfy the INP. From Lemma~\ref{lem:recolJoin} it is easily seen that the INP is preserved in that case. We now argue that the resulting graph $H$ is a forest. Suppose to the contrary that $H$ contains a cycle $C=(u_0,v_0),(u_1,v_1),\ldots,(u_k,v_k)$ with $u_0=u_k$ and $v_0=v_k$ (we represent nodes of $H$ by tuples $(x,y)$ where $x\in V(H_1)$ and $y\in V(H_2)$, as shown in Lemma~\ref{lem:recolJoin}). 
Then $u_0,u_1,\ldots,u_k$ is a closed walk in $H_1$, and $v_0,v_1,\ldots,v_k$ is a closed walk in $H_2$, of length $k\ge 3$. 
Since $H_2$ is a forest, there is an index $i$ such that $v_{i-1}=v_{i+1}$.
It follows that $\ell_1(u_{i-1})=\ell_2(v_{i-1})=\ell_2(v_{i+1})=\ell_1(u_{i+1})$. But $u_{i-1}$ and $u_{i+1}$ are both neighbors of $u_i$, so since $H_1$ satisfies the INP, $u_{i-1}=u_{i+1}$. We conclude that the vertices $(u_{i-1},v_{i-1})$ and $(u_{i+1},v_{i+1})$ in $H$ are the same, contradicting that $C$ is a cycle. So $H$ is a forest that satisfies the INP.
\qed
\end{proof}

Combining the above lemmas yields:

\begin{thm}
\label{thm:invariant}
Let $k\ge 3$. Let $G=(V,E)$ be a $(k-2)$-connected $k$-colorable chordal graph, and let $T\subseteq V(G)$ be a clique of $G$ with $m=|T|\ge k-2$. Then $\C^c_k(G,T)$ is an $(m,k)$-color-complete graph, or it is a forest that satisfies the INP.
\end{thm}

\begin{proof}
We prove the statement by induction over $2|V|-|T|$. 
If $T=V(G)$, then $\C^c_k(G,T)$ is isomorphic to $\C_k(G)$, with the trivial label function (Lemma~\ref{propo:recolLeaf}), so this is an $(m,k)$-color-complete graph (since $T$ is a clique). Now assume that $T\not=V(G)$. 

If $(G,T)$ can be obtained from a graph $(G,T\cup \{v\})$ using a forget operation, where $T\cup \{v\}$ is a clique of $G$, then by induction, $\C^c_k(G,T\cup \{v\})$ is either an $(m+1,k)$-color-complete graph or a forest that satisfies the INP. Because $T\cup \{v\}$ is a clique on $m+1$ vertices and $G$ is $k$-colorable, $m\le k-1$. If $m=k-1$ then $\C^c_k(G,T\cup \{v\})$ is a set of isolated nodes. This shows that Lemma~\ref{lem:invariantForget} covers all cases, and thus $\C^c_k(G,T)$ satisfies the desired property. 

If $(G,T)$ can be obtained from a graph $(G-v,T\bs\{v\})$ using an introduce operation then
Lemma~\ref{propo:connectedness} shows that $G-v$ is  $(k-2)$-connected, and obviously it is chordal, so 
we may use induction to conclude that
$\C^c_k(G-v,T\bs \{v\})$ is either an $(m+1,k)$-color-complete graph or a forest that satisfies the INP.
Lemma~\ref{propo:connectedness} also shows that  $|T|\ge k-1$. 
This shows that Lemma~\ref{lem:invariantIntroduce} covers all cases, and thus $\C^c_k(G,T)$ satisfies the desired property. 

In the remaining case, Lemma~\ref{lem:ChordalNiceTreedecomp} shows that $(G,T)$ is the join of two (smaller) graphs $(G_1,T)$ and $(G_2,T)$, 
which are  $(k-2)$-connected (Lemma~\ref{propo:connectedness}), and chordal since they are induced subgraphs of $G$, so we can use the 
induction 
hypothesis.
Then Lemma~\ref{lem:invariantJoin} can be applied, to show that $\C^c_k(G,T)$ satisfies the desired property. 
\qed
\end{proof}

\noindent
{\bf Remark 2.}
The examples in Figure~\ref{fig:cutvertCSG} show that if we relax the connectivity requirement to $(k-3)$-connectedness, 
the property in Theorem~\ref{thm:invariant} 
does not necessarily hold
anymore:
the examples in Figure~\ref{fig:cutvertCSG}(c) and~(d) are not forests, and the example in Figure~\ref{fig:cutvertCSG}(e) does not satisfy the INP. 
Hence, we cannot generalize our polynomial-time result on ${\cal C}_k$-{\sc Reachability} to $(k-3)$-connected chordal graphs in a straightforward way.

\medskip
\noindent
The characterization 
of $\C^c_k(G,T)$ in Theorem~\ref{thm:invariant}
does not yet guarantee that simply keeping track of the (relevant component of the) \CSG\  yields a polynomial-time algorithm, as shown by the second example in Section~\ref{sec:badexamples}.
However, we will now show that it suffices to only keep track of the following essential information, which remains polynomially bounded. 

\subsection{An Efficient Algorithm: Computing Essential Information}
\label{ssec:algorithm}

Let $G=(V,E)$ be a graph with $T\subseteq V$, and let $\alpha$ and $\beta$ be $k$-colorings of a supergraph of $G$. 
(The graph $G$ should be viewed as a subgraph that occurs during the dynamic programming, while $\alpha$ and $\beta$ are the colorings of the full graph.)
Let $\alpha'=\alpha\restr_{V}$ and $\beta'=\beta\restr_{V}$. 
If $\C^c_k(G,T)$ is a forest with the $\alpha'$-node $x$ and $\beta'$-node $y$ in the same component, then we define the {\em $\alpha$-$\beta$-path} to be the unique path in $\C^c_k(G,T)$ with end vertices $x$ and $y$ (together with its vertex labels). 
Given the two colorings $\alpha$ and $\beta$, the {\em essential information} for $\C^c_k(G,T)$ consists of the following:
\begin{itemize}
\item 
whether the $\alpha'$ and $\beta'$ nodes appear in the same component, 
\item 
whether $\C^c_k(G,T)$ is a forest, and 
\item 
in case the answers to both questions are yes: the $\alpha$-$\beta$-path in $\C^c_k(G,T)$ (including vertex labels). 
\end{itemize}

We also need to prove a polynomial upper bound on the length of the $\alpha$-$\beta$-path. This is nontrivial, since the introduce operation may increase the length by a factor~2. However, we will show that this only happens when earlier, a forget operation has decreased the length by a similar amount. To formalize this, we use the following alternative length measure for paths in CSGs for recoloring. 

For a subgraph $F$ of $G$ and $v\in V(G)$, we denote the neighbors of $v$ in~$F$ by $N_F(v)=N(v)\cap V(F)$.
Let $(H,\ell)$ be a labeled graph, where every node label~$\ell(v)$ is a $k$-coloring of a complete graph on vertex set $T$. The set of colors {\em used by} a node $v\in V(H)$  is defined as $U(v)=\{\ell(v)(x) \mid x\in T\}$. If $P$ is a subgraph in $H$ and $v\in V(P)$, then the {\em node weight} for $v$ is defined as $w_P(v)=|(\cup_{x\in N_P(v)} U(x))\bs U(v)|$. So this is the total number of colors that are used in the labels (colorings) for neighbors of $x$ in $P$, that are not used by the label for $x$ itself. 
We define the {\em weight of a subgraph $P$ of $H$} to be $w(P)=\sum_{v\in P} w_P(v)$.
For example, consider the last CSG shown in Figure~\ref{fig:unitintCSG}: the vertex with label $24$ has 
weight~1 
in the path with node labels $32,34,24,23,21$, but 
weight~2
in the  path with node labels $32,34,24,14,12$. This weight depends on whether the corresponding path in the previous CSG (before forgetting $f$) contained the (blue) edge between nodes $124$ and $324$.
The main idea is that for a path $P$, $w(P)$ bounds the length of $P$.
This follows from the following simple observation, where we deduce amongst others that $w_P(v)\geq 1$ if $v$ is not an isolated vertex.

\begin{propo}
\label{propo:weightbounds}
Let $(H,\ell)$ be a labeled graph, where every node label~$\ell(v)$ is a $k$-coloring of a complete graph on a vertex set $T$, such that adjacent nodes do not have the same label. 
Then for any subgraph $P$ of $H$ and any vertex $v\in V(P)$: $1\le w_P(v)\le k-|T|$ if $v$ is not an isolated vertex in $P$, and $w_P(v)=0$ otherwise.
\end{propo}

We observe that, as soon as the $\alpha'$ and $\beta'$ nodes are separated in some CSG that occurs during the dynamic programming, we may terminate and return 
NO. 

\begin{propo}
\label{propo:StopWhenNO}
Let $G'=(V',E')$ be a subgraph of $G=(V,E)$, and let $\alpha$ and $\beta$ be two $k$-colorings of $G$. Let $\alpha'=\alpha\restr_{V'}$ and $\beta'=\beta\restr_{V'}$. 
For any $T'\subseteq V'$ and $T\subseteq V$: if the $\alpha'$ and $\beta'$ nodes of $\C^c_k(G',T')$ are separated, then the $\alpha$ and $\beta$ nodes of $\C^c_k(G,T)$ are separated. 
\end{propo}

\begin{proof}
Suppose that the $\alpha$ and $\beta$ nodes of $\C^c_k(G,T)$ are not separated. Then by Lemma~\ref{propo:one}, there exists a recoloring sequence $\gamma_0,\ldots,\gamma_p$ from $\alpha$ to $\beta$. Then restricting all of these colorings to $V'$ yields a recoloring sequence 
$\gamma_0\restr_{V'},\ldots,\gamma_p\restr_{V'}$ from $\alpha'$ to $\beta'$ for $G'$. So using Lemma~\ref{propo:one} again, the $\alpha'$ and $\beta'$ nodes in $\C^c(G',T')$ are not separated.
\qed
\end{proof}

Note that in the next lemmas, `polynomial time' means polynomial in the entire input size, which includes the essential information; in particular, the path length (we will show later, namely in the proof of Theorem~\ref{thm:main}, that the maximum path length that can occur is at most $2(k+3)n$).

\begin{lem}
\label{lem:pathForget}
Let $G$ be a $(k-2)$-connected $k$-colorable chordal graph and let $T$ be a clique of $G$ with $k-1\le |T|$, and $v\in T$. 
If we know the essential information for $\C^c_k(G,T)$, then in polynomial time we can compute the essential information for $\C^c_k(G,T\bs \{v\})$. 
If $\C^c_k(G,T)$ has a unique $\alpha$-$\beta$-path $P$, then $\C^c_k(G,T\bs \{v\})$ has a unique $\alpha$-$\beta$-path $P'$, and $w(P')\le w(P)$.
\end{lem}

\begin{proof}
Theorem~\ref{thm:invariant} shows that $\C^c_k(G,T)$ is either an $(m,k)$-color complete graph or a forest that satisfies the INP. 
Lemma~\ref{lem:invariantForget} then shows that $\C^c_k(G,T\bs \{v\})$ is a forest if and only if $\C^c_k(G,T)$ is a forest. Proposition~\ref{propo:StopWhenNO} shows that if $\C^c_k(G,T)$ has no $\alpha$-$\beta$-path, then $\C^c_k(G,T\bs \{v\})$ has no $\alpha$-$\beta$-path. 
If $\C^c_k(G,T)$ is a forest with a unique $\alpha$-$\beta$-path $P$, then Lemma~\ref{lem:recolForget} shows that we can find an $\alpha$-$\beta$-path $P'$ in $\C^c_k(G,T\bs \{v\})$ by starting with $P$, adjusting the labels, and possibly contracting some edges. This yields the unique $\alpha$-$\beta$-path in the forest $\C^c_k(G,T\bs \{v\})$.

We go more into detail on the construction of $P'$ from $P$ in order to prove that $w(P')\le w(P)$. If $|T|=k$ then $\C^c_k(G,T)$ consists of only isolated nodes, and thus $\C^c_k(G,T\bs \{v\})$ (which is a contraction of the former graph) as well, so the statement is trivial. 
So now assume that $|T|=k-1$. 
By Proposition~\ref{propo:weightbounds}, every node in $P$ has weight~1, and nodes in $P'$ have weight at most~2. 
So to prove that $w(P')\le w(P)$, it suffices to show that every node of $P'$ with 
weight~2 
results from contracting an edge of $P$ (that is, contracting two nodes of weight~1). 
Denote by~$\ell$ the node labels in $\C^c_k(G,T\bs \{v\})$ (which are $(k-2)$-colorings of $G[T\bs \{v\}]$).
Consider a node $y\in V(P')$ with weight~2, so it has two neighbors $x,z\in V(P')$. 
Let $a\in U(x)\bs U(y)$, and $b\in U(z)\bs U(y)$. 
Since $w_{P'}(y)=2$, it holds that $a\not=b$, so $U(x)\cup U(z)=\{1,\ldots,k\}$. 
So it is not possible to extend $\ell(x)$, $\ell(y)$ and $\ell(z)$ to 
$k$-colorings of $G[T]$ by assigning the same color to $v$, and therefore the node $y$ resulted from contracting two nodes of $P$.
\qed
\end{proof}

\begin{lem}
\label{lem:pathIntroduce}
Let $G=(V,E)$ be a $(k-2)$-connected $k$-colorable chordal graph. Let $T$ be a clique of $G$ with $T\not=V$, such that $(G,T)$ can be obtained from $(G-v,T\bs \{v\})$ by using an introduce operation. 
Let $\alpha$ and $\beta$ be two $k$-colorings of a supergraph of $G$.
If we know the essential information for $\C^c_k(G-v,T\bs \{v\})$, then in polynomial time we can compute the essential information for $\C^c_k(G,T)$. 
If $\C^c_k(G,T)$ has a unique $\alpha$-$\beta$-path $P'$, then $w(P')\le w(P)+2$ in the case where $\C^c_k(G-v,T\bs \{v\})$ has a unique $\alpha$-$\beta$-path $P$ and $w(P')=0$ otherwise.
\end{lem}

\begin{proof}
Let $(H,\ell)=\C^c_k(G-v,T\bs \{v\})$ and $(H',\ell')=\C^c_k(G,T)$, such that $(H',\ell')$ is obtained from $(H,\ell)$ as shown in Lemma~\ref{lem:recolIntroduce}.
Let $\alpha'=\alpha\restr_{V(G)}$ and $\beta'=\beta\restr_{V(G)}$. 
By Lemma~\ref{propo:connectedness} we find that $G-v$ is $(k-2)$-connected. Hence, we can use 
Theorem~\ref{thm:invariant} to deduce that $(H,\ell)$ is either an $(m,k)$-color-complete graph (for $m=|T|-1$), or a forest that satisfies the INP. 

By Lemma~\ref{propo:connectedness} we find that $|T|=k$ or $|T|=k-1$.
First suppose that $|T|=k$. Then $H'$ is a forest consisting of isolated nodes. So its $\alpha'$ and $\beta'$ nodes are in the same component if and only if they are the same. This holds if and only if $\alpha'\restr_T=\beta'\restr_T$, and either $(H,\ell)$ is a $(k-1,k)$-color-complete graph or a forest with an $\alpha$-$\beta$-path of length~0. Clearly the $\alpha$-$\beta$-path in $H'$ has length~0 in this case, and the label of its node is $\alpha\restr_T$. 
This shows that we can deduce, in polynomial time, the essential information for $H'$ if $|T|=k$.

Now suppose that $|T|=k-1$. Recall that  $(H,\ell)$ is either a $(k-2,k)$-color-complete graph or a forest that satisfies the INP. 
Then by Lemma~\ref{lem:invariantIntroduce} the following holds. If $(H,\ell)$ is  a $(k-2,k)$-color-complete graph, then
$(H',\ell')$ is a $(k-1,k)$-color-complete graph.
If $H$ is a forest that satisfies the INP, then $H'$ is a forest that satisfies the INP.
Hence, $H'$ is a forest if and only if $H$ is a forest.
Below we show how to deduce in polynomial time the essential information for $H'$.

Proposition~\ref{propo:StopWhenNO} shows that if $H$ has no $\alpha$-$\beta$-path, then $H'$ has no $\alpha$-$\beta$-path. 
Hence it remains to consider the case where $H$ has a unique $\alpha$-$\beta$-path $P$. We will apply Lemma~\ref{lem:recolIntroduce} to the nodes of $P$ to construct a (labeled) $\alpha$-$\beta$-path $P'$, which is a (labeled) subgraph of $(H',\ell')$, and thus the unique $\alpha$-$\beta$-path in $(H',\ell')$, and show that $w(P')\le w(P)+2$. 
(As an illustration of this proof, consider for instance how in Figure~\ref{fig:unitintCSG}, the path $P'$ with $w(P')=7$ between node labels $142$ and $312$ in the CSG with $T=(f,g,h)$ is deduced from the path $P$ with $w(P)=5$ between node labels $14$ and $31$ in the previous CSG.)

Since $|T\bs \{v\}|=k-2$, Lemma~\ref{lem:recolIntroduce} shows that every node $x\in V(P)$ yields two adjacent nodes of $H'$, which we will denote as  $x_1$ and $x_2$, such that the labels $\ell'(x_1)$ and $\ell'(x_2)$ 
($(k-1)$-colorings 
of $G[T]$) assign the two colors $a$ and $b$ that are not used by $\ell(x)$ to vertex $v$. 
For any two adjacent nodes $x$ and $y$ in $P$, $\ell(x)$ and $\ell(y)$ differ on exactly one vertex of $T\bs \{v\}$, and both are $(k-2)$-colorings, so there is at least one color $c$ that is used neither by $\ell(x)$ nor by $\ell(y)$. 
So we can choose indices $i,j\in \{1,2\}$ such that $\ell'(x_i)(v)=c$ and $\ell'(y_j)(v)=c$, and therefore $x_i$ and $y_j$ are adjacent in $H'$ (Lemma~\ref{lem:recolIntroduce}). 
These two observations show that if we take the two nodes $x_1$ and $x_2$ for every $x\in V(P)$, then all of these nodes together induce a connected subgraph of $H'$ (a caterpillar with a perfect matching in fact) that contains the $\alpha'$-node and the $\beta'$-node of $H'$. Within this subgraph we can easily find the new $\alpha$-$\beta$-path $P'$.
Every node of the new path $P'$ has weight~1 (Proposition~\ref{propo:weightbounds}). The total weight of $P$ may increase by~2 if both end nodes of $P$ are replaced by a pair of nodes this way. Nevertheless, we will now show that the weight cannot increase by more than~2, by showing that internal nodes $y$ of $P$ are only replaced by a pair of nodes $y_1$ and $y_2$ in $P'$ if $w_{P}(y)=2$. 

Consider a node $y\in V(P)$ with neighbors $x$ and $z$ on $P$, such that 
without loss of generality
the path $P'$ contains the new nodes $x_1,y_1,y_2,z_1$, in this order. Let $u\in T\bs \{v\}$ be the unique vertex that the colorings $\ell'(x_1)$ and $\ell'(y_1)$ differ on, and let $w\in T\bs \{v\}$ be the unique vertex that the colorings $\ell'(y_2)$ and $\ell'(z_1)$ differ on. 
Since all of the colorings $\ell'(x_1),\ell'(y_1),\ell'(y_2),\ell'(z_1)$ use $k-1$ colors out of a total of $k$ colors, we conclude that $\ell'(y_2)(v)=\ell'(x_1)(u)$, and similarly, $\ell'(y_1)(v)=\ell'(z_1)(w)$. 
It follows that the colorings $\ell(x)=\ell'(x_1)\restr_{T\bs \{v\}}$ and $\ell(z)=\ell'(z_1)\restr_{T\bs \{v\}}$ together still use all $k$ colors, and therefore $w_{P}(y)=2$. We conclude that internal nodes of $P$ cannot contribute a weight increase, so $w(P')\le w(P)+2$.
\qed
\end{proof}

\begin{lem}
\label{lem:pathJoin} 
Let $G=(V,E)$ be a $(k-2)$-connected $k$-colorable chordal graph and let $T$ be a clique of $G$, such that $(G,T)$ can be obtained from $(G_1,T)$ and $(G_2,T)$ using a join operation.
If we know the essential information for both $\C^c(G_1,T)$ and $\C^c(G_2,T)$, then in polynomial time we can compute the essential information for $\C^c(G,T)$. 
If $\C^c_k(G,T)$ is a forest with a unique $\alpha$-$\beta$-path $P$, then for at least one choice of $i\in \{1,2\}$, $\C^c(G_i,T)$ is a forest with a unique $\alpha$-$\beta$-path $P_i$, and $w(P)=w(P_i)$. 
\end{lem}

\begin{proof}
Let $(H,\ell)=\C^c_k(G,T)$, $(H_1,\ell_1)=\C^c_k(G_1,T)$ and $(H_2,\ell_2)=\C^c_k(G_2,T)$ be labeled graphs such that $(H,\ell)$ is obtained from $(H_1,\ell_1)$ and $(H_2,\ell_2)$ as shown in Lemma~\ref{lem:recolJoin}. 
By Theorem~\ref{thm:invariant}, for $i\in \{1,2\}$, $(H_i,\ell_i)$ is either a $(|T|,k)$-color complete graph  or a forest that satisfies the INP 
(Lemma~\ref{propo:connectedness} shows that the graphs are $(k-2)$-connected).

By Lemma~\ref{lem:invariantJoin}, $H$ is a forest (that satisfies the INP) if and only if at least one of $H_1$ and $H_2$ is a forest. By Proposition~\ref{propo:StopWhenNO}, if there is no $\alpha$-$\beta$-path in one of $H_1$ and $H_2$, then there is no $\alpha$-$\beta$-path in $H$. So now assume that both $H_1$ and $H_2$ contain an  $\alpha$-$\beta$-path (though possibly not unique). 
If one of these, say $H_i$, is a forest with a unique $\alpha$-$\beta$-path $P_i$, but the other is a color-complete graph, then the unique $\alpha$-$\beta$-path $P$ of $H$ is 
the same as
$P_i$ (Lemma~\ref{lem:invariantJoin}), and thus $w(P)=w(P_i)$. 

It only remains to consider the case that both $H_1$ and $H_2$ are forests and contain a unique $\alpha$-$\beta$-path; call these $P_1$ and $P_2$ respectively. If 
$P_1$ equals $P_2$,
then $H$ also has an $\alpha$-$\beta$-path that 
equals
these paths (Lemma~\ref{lem:recolJoin}), which is therefore the unique $\alpha$-$\beta$-path $P$ in $H$, with $w(P)=w(P_1)=w(P_2)$. 
We conclude the proof by showing the other direction. 
(This is similar to the last part of the proof of Lemma~\ref{lem:invariantJoin}.)
Suppose $H$ has an $\alpha$-$\beta$-path $P=v_0,\ldots,v_p$. Every node $v_i$ of $H$ corresponds to a pair $v^1_i$ and $v^2_i$ of nodes in $H_1$ resp.\ $H_2$, with $\ell(v_i)=\ell_1(v^1_i)=\ell_2(v^2_i)$, and $P_1=v^1_0,\ldots,v^1_p$ and $P_2=v^2_0,\ldots,v^2_p$ are $\alpha$-$\beta$-walks in $H_1$ resp.\ $H_2$ (Lemma~\ref{lem:recolJoin}). 
If one of these, say $P_1$, is not a path, then since $H_1$ is a forest, there exists an index $i$ such that $v^1_{i-1}=v^1_{i+1}$. So $\ell(v_{i-1})=\ell_1(v^1_{i-1})=\ell_1(v^1_{i+1})=\ell(v_{i+1})$. 
But since $P$ is a path, $v_{i-1}$ and $v_{i+1}$ are distinct neighbors of $v_i$, so this contradicts the INP. 
We conclude that both $P_1$ and $P_2$ are paths, 
so $P_1$, $P_2$ and $P$ are all equal, 
so $w(P)=w(P_1)=w(P_2)$. 
This concludes the proof, which shows that we can decide in polynomial time whether $H$ is a forest with an $\alpha$-$\beta$-path, and compute it in that case. 
\qed
\end{proof} 

Combining the above statements yields the main result of this section:

\begin{thm}
\label{thm:main}
Let $G$ be a $k$-colorable $(k-2)$-connected chordal graph, and let $\alpha$ and $\beta$ be two $k$-colorings of $G$. Then in polynomial time, we can decide whether $\C_k(G)$ contains an $\alpha$-$\beta$ path. 
\end{thm}

\begin{proof}
Corollary~\ref{corol:chordalNiceTreedecomp} shows that for every chordal $k$-colorable graph $G$ on $n$ vertices, we can find in polynomial time a chordal nice tree decomposition on at most $(k+3)n$ nodes. 
So every node of this tree decomposition corresponds to a $(k-2)$-connected chordal subgraph $H$ of $G$ with terminal set $T$, such that either $H$ is a clique with $T=V(H)$ (leaf nodes), or $(H,T)$ can be obtained from the graph(s) corresponding to its child node(s) using a forget, introduce or join operation. 
(The fact that all of these graphs are $(k-2)$-connected follows inductively using Lemma~\ref{propo:connectedness}, and that they are chordal follows since they are induced subgraphs.)
For every one of those terminal subgraphs, we compute the essential information, bottom up (Lemma~\ref{propo:invariantLeaf}, Lemmas~\ref{lem:pathForget}--\ref{lem:pathJoin}). The computation terminates, answering NO, as soon as one subgraph $(H,T)$ is encountered such that $\alpha$ and $\beta$ are separated in $\C^c_k(H,T)$, which is correct by Proposition~\ref{propo:StopWhenNO}. (We remark that this can occur when $(H,T)$ is obtained by a join operation, or by an introduce operation when $|T|=k$.)
Otherwise, the computation terminates for the root node of the tree decomposition, which corresponds to the entire graph $G$ itself, with some terminal set $T$, with the conclusion that either $\C^c_k(G,T)$ is a color-complete graph, or that it is a forest that contains an $\alpha$-$\beta$-path. In either case, the answer to the problem is YES (Lemma~\ref{propo:one}).

Now we consider the complexity. We find the chordal nice tree decomposition in polynomial time, and it has at most $(k+3)n$ nodes (Corollary~\ref{corol:chordalNiceTreedecomp}). 
Computing the essential information for all nodes can be done in polynomial time, although the input size here includes the $\alpha$-$\beta$-path. Nevertheless, every operation increases the weight of the path by at most~2 (Lemmas~\ref{lem:pathForget}, \ref{lem:pathIntroduce} and \ref{lem:pathJoin}), and in every case the weight of the path is an upper bound for its length (Proposition~\ref{propo:weightbounds}), so the maximum path length that can occur during the execution of the algorithm is at most $2(k+3)n$. 
Together this shows that the whole procedure terminates in polynomial time. 
\qed
\end{proof}

We stress that $(m,k)$-color complete graphs, which have $k!/(k-m)!$ nodes, are not computed explicitly in our algorithm. 
So indeed,  in order to obtain a polynomial-time algorithm,
we do not need to assume that $k$ is a constant.

\section{Discussion}\label{s-discussion}

Due to the PSPACE-completeness result  of Hatanaka, Ito and Zhou~\cite{HIZ17b} on {\sc $\C_k$-Reachability} for chordal graphs, which holds when $k$ is a sufficiently large constant,
our polynomial-time algorithm cannot be extended to all chordal graphs.
Since {\sc $\C_k$-Reachability} problem is polynomial-time solvable for general graphs if $k=3$~\cite{CHJ11} and PSPACE-complete for $k=4$~\cite{BC09}, it would nevertheless still be interesting  to determine the complexity of {\sc $\C_4$-Reachability} for chordal graphs (with at least one cut vertex).
We refer to Remark~2 for a brief discussion on why our current proof technique does not work for this case. We also note that the complexity of {\sc $\C_4$-Reachability} is open for proper interval graphs. 
Initial experimental results suggest that solving this problem is not straightforward.

Below we discuss a number of other possible directions for future work, which may require additional experimental results in order to refine our DP method.
The two most important research goals are the following:

\medskip
\noindent
{\it 1. 
Explore for which other solution graph concepts $\SSS$
the DP method 
can be used to obtain polynomial-time algorithms
for the 
{\sc $\SSS$-Reachability} problem.} 

\medskip
\noindent
{\it 2. 
Explore which other commonly studied reconfiguration problems can be solved 
efficiently 
using CSGs.}

\medskip
\noindent 
The method of using CSGs can be applied to solve the {\sc $\SSS$-Connectivity} problem. Hence, this problem (which is one of the central problems in reconfiguration) is a most suitable candidate problem for the second research goal.
In this context we recall that the {\sc $\C_k$-Connectivity} problem is trivial for chordal graphs~\cite{BJLPP14} (see Section~\ref{sec:intro}).
Nevertheless, studying the complexity of the following related problem seems interesting. Call two $k$-colorings $\alpha$ and $\beta$ of a graph $G$ {\em compatible} if they coincide on all $k$-cliques of $G$. Given a chordal graph $G$ and $k$-coloring $\alpha$, is the subgraph of $\C_k(G)$ induced by all $k$-colorings that are compatible with $\alpha$ connected? 

\medskip
\noindent
Finally we discuss 
the  list coloring generalization $\C_L$ of $\C_k$.
In Remark~1, we explained how to
 generalize the DP rules presented in Section~\ref{sec:DP} to $\C_L$ (namely, by simply omitting all nodes that correspond to invalid vertex colors).
In this way, 
we showed that 
the DP rules presented in~\cite{HIZ15} can be generalized.
However, it is 
not obvious whether the results from Section~\ref{sec:chordalgraphs} also generalize to list colorings.
The following question by Hatanaka (asked at CoRe 2015) is also interesting:
is there a polynomial-time algorithm for {\sc $\C_L$-Reachability} restricted to trees?
Note that {\sc ${\C_k}$-Reachability} is trivial for trees, because $\C_k(G)$ is connected for every tree~$G$ and every integer $k\geq 3$ 
(see~\cite{BJLPP14}; this also follows easily from Lemma~\ref{propo:degeneracy}).

\medskip
\noindent
{\bf Acknowledgement.} The authors would like to thank Carl Feghali and Matthew Johnson for 
fruitful discussions and three anonymous reviewers for comments on a previous version of our paper.

\bibliographystyle{abbrv}

\appendix

\section{The Proofs of Lemmas~\ref{lem:recolForget}--\ref{lem:recolJoin}}\label{a-rules}

\def\Cnew{H'}  
\def\Corig{H}  

\medskip
\noindent
{\bf Lemma~\ref{lem:recolForget} (Forget) [Restated].}
{\it Let $(G,T)$ be a terminal graph.
For every $v\in T$, 
it holds that
$(H',\ell')=\C^c_k(G,T\bs \{v\})$ can be computed from $(H,\ell)=\C^c_k(G,T)$ as follows:
\begin{itemize}
 \item For every node $x$ in $\Corig$ with $\ell(x)=\gamma$, let
 $\ell'(x)=\gamma\restr_{T\bs\{v\}}$.
 \item Iteratively contract every edge between two nodes $x$ and $y$ with $\ell'(x)=\ell'(y)$ and assign
 label $\ell'(z):=\ell'(x)$ to the resulting node $z$. 
 \end{itemize}
Moreover, for any coloring $\gamma$ of $G$, the $\gamma$-node of $\C^c_k(G,T\bs \{v\})$ is the node that results from contracting the set of nodes that includes the $\gamma$-node of $\C^c_k(G,T)$.}

\begin{proof}
Let $S$ denote the certificate for $(\Corig,\ell)$, so for every node $x$ of $\Corig$, $S_x$ denotes the set of $k$-colorings of $G$ (or {\em solutions}), such that these sets satisfy the properties stated in Lemma~\ref{lem:characterizationCSGs}. In addition, for every coloring $\gamma$ for which a $\gamma$-node $x$ has been marked in $\Corig$, we may assume that $\gamma\in S_x$. 
We will prove the statement using Lemma~\ref{lem:characterizationCSGs} again, by giving a certificate $S'$ for $(\Cnew,\ell')$, and proving that the five properties hold for these. 

The graph $\Cnew$ is obtained by iteratively contracting edges of $\Corig$, so every node 
$y$ of $\Cnew$ corresponds to a connected set of nodes of $\Corig$, which we will denote by $M_y$. So $\{M_y \mid y\in V(\Cnew)\}$ is a partition of $V(\Corig)$. For every node $y\in V(\Cnew)$, we define $S'_y=\cup_{x\in M_y} S_x$. 

For every $k$-coloring $\gamma$ of $G$ such that the $\gamma$-node $x\in V(\Corig)$
is marked, we define the 
$\gamma$-node of $\Cnew$ to be the node~$y$ with $x\in M_y$. Clearly, $\gamma\in S'_y$ then holds, so this is correct.
It now remains to verify that the solution sets $S'_x$ satisfy the five properties stated in Lemma~\ref{lem:characterizationCSGs}.

\begin{enumerate}[1.]
\item 
$\{S_x \mid x\in V(\Corig)\}$ is a partition of the nodes of $\C_k(G)$ (Lemma~\ref{lem:characterizationCSGs}(\ref{pr:partition})), and $\{M_y \mid y\in V(\Cnew)\}$ is a partition of $V(\Corig)$, so $\{S'_y \mid y\in V(\Cnew)\}$ is a partition of the nodes $\C_k(G)$.
 
\item 
Consider a node $y\in V(\Cnew)$, with label $\ell'(y)$, which is a $k$-coloring of $G[T\bs \{v\}]$. Every node $x\in M_y$ has a label $\ell(x)$ with $\ell(x)\restr_{T\bs \{v\}}=\ell'(y)$, and for every $\gamma\in S_x$, it holds that $\gamma\restr_{T}=\ell(x)$ (Lemma~\ref{lem:characterizationCSGs}(\ref{pr:correctlabels})), and thus  $\gamma\restr_{T\bs \{v\}}=\ell'(y)$. 
Therefore, for every 
$\gamma\in S'_y$, it holds that $\gamma\restr_{T\bs \{v\}}=\ell'(y)$.

\item 
Consider two adjacent nodes $x$ and $y$ in $\Cnew$. This implies that there exists an edge $ab$ between the node sets $M_x$ and $M_y$ of $\Corig$. By definition, all nodes $a\in M_x$ have $\ell'(a)=\ell'(x)$, and all nodes $b\in M_y$ have $\ell'(b)=\ell'(y)$. So if $\ell'(x)=\ell'(y)$, then the edge $ab$ should also have been contracted when constructing $\Cnew$, a contradiction. Hence $\ell'(x)\not=\ell'(y)$.

\item 
Consider a node $x$ of $\Cnew$. The node set $M_x$ is connected, so for any two nodes $y,z\in M_x$, the subgraph of $\Corig$ induced by $M_x$ contains a path from $y$ to $x$. Edges $ab$ of this path correspond to 
solution sets $S_a$ and $S_b$ that contain adjacent solutions (Lemma~\ref{lem:characterizationCSGs}(\ref{pr:adjacency})). In addition, all such solution sets~$S_a$ are connected in $\C_k(G)$ (Lemma~\ref{lem:characterizationCSGs}(\ref{pr:connectedsets})). Combining these facts shows that the new solution sets $S'_x$ are connected in $\C_k(G)$. 

\item 
Let $z$ and $z'$ be two nodes of $\Cnew$. By construction, $z$ and $z'$ are adjacent if and only if there exist nodes $x\in M_z$ and $x'\in M_{z'}$ that are adjacent in $\Corig$.
Two such nodes $x$ and $x'$ are adjacent in $\Corig$ if and only if there exist solutions $\alpha\in S_x$ and $\alpha'\in S_{x'}$ that are adjacent in $\C_k(G)$ (Lemma~\ref{lem:characterizationCSGs}(\ref{pr:adjacency})). 
Using the definition of $S_z$ and $S_{z'}$, we conclude that $z$ and $z'$ are adjacent if and only if there exist solutions $\alpha\in S_z$ and $\alpha'\in S_{z'}$ that are adjacent in $\C_k(G)$.\qed
\end{enumerate}
\end{proof}

\def\Cnew{H'}	
\def\Corig{H}

\medskip
\noindent
{\bf Lemma~\ref{lem:recolIntroduce} (Introduce) [Restated].}
{\it Let $(G,T)$ be a terminal graph obtained from a terminal graph $(G-v,T\setminus \{v\})$ by introducing $v$.
Then $(\Cnew,\ell')=C^c_k(G,T)$ can be computed as follows from $(\Corig,\ell)=\C^c_k(G-v,T\setminus \{v\})$:
\begin{itemize}
 \item For every node $x$ of $\Corig$ with label~$\ell(x)$, and 
 every color $c\in \{1,\ldots,k\}$: if the (unique) function $\delta:T\to \{1,\ldots,k\}$ with $\delta(v)=c$ and $\delta\restr_{T}=\ell(x)$ is a 
 coloring of $G[T]$ then introduce a node $x_c$ with label $\ell'(x_c)=\delta$.
 \item 
 For every pair of distinct nodes $x_c$ and $y_d$: add an edge between them if and only if
 (1) $x=y$ or (2) $xy$ is an edge in $\Corig$ and $c=d$. 
\end{itemize} 
Moreover, for every $k$-coloring $\gamma$ of $G$, if $x$ is the $\gamma\restr_{V(G)\bs \{v\}}$-node in $\Corig$ and $\gamma(v)=c$, then $x_c$ is the $\gamma$-node of $\Cnew$.}

\begin{proof}
Let $S$ be a certificate for $(\Corig,\ell)$, so for every node $x$ of $\Corig$, let $S_x$ denote the set of $k$-colorings of $G-v$ 
(or solutions), 
such that these sets satisfy the properties stated in Lemma~\ref{lem:characterizationCSGs}. 
In addition, for every coloring $\gamma$ for which a $\gamma$-node $x$ has been marked in $\Corig$, we may assume that $\gamma\in S_x$. 
Now we construct a certificate $S'$ for $(\Cnew,\ell')$. 
For every node $x_c$ of $\Cnew$ (that corresponds to a node~$x$ of $\Corig$, and to assigning a color $c$ to the new vertex $v$), we define $S'_{x_c}$ to be the set of $k$-colorings $\alpha$ of $G$ with $\alpha(v)=c$ and $\alpha\restr_{T\bs \{v\}}\in S_x$. 
For every $k$-coloring~$\gamma$ of $G$ and node $x_c$ of $\Cnew$, we define $x_c$ to be the $\gamma$-node of $\Cnew$ if and only if $\gamma(v)=c$ and $x$ is the $\gamma\restr_{V(G)\bs \{v\}}$-node of $\Corig$. 
Clearly, this guarantees $\gamma\in S'_{x_c}$ for the chosen $\gamma$-node $x_c$.
To prove the statement, it only remains to show that the new solution sets $S'_{x_c}$ satisfy the five properties stated in Lemma~\ref{lem:characterizationCSGs}.

\begin{enumerate}[1.]
\item 
First, we observe that for every node $x_c$ of $\Cnew$, $S'_{x_c}$ is a nonempty set of $k$-colorings of $G$, because $S_x$ is nonempty (Lemma~\ref{lem:characterizationCSGs}), and by choice of $c$, every coloring~$\alpha\in S_x$ can be extended to a 
$k$-coloring of $G$ by setting $\alpha(v)=c$ (this uses the fact that $N(v)\subseteq T$). 
So to prove that the new solution sets form a partition of the nodes of $\C_k(G)$, it only remains to show that every $k$-coloring~$\alpha$ of $G$ is included in $S'_{x_c}$ for exactly one new node $x_c$. 
For every such $\alpha$, there exists a unique node $x$ of $\Corig$ such that $\alpha\restr_{V(G)\bs \{v\}}\in S_x$ (Lemma~\ref{lem:characterizationCSGs}(\ref{pr:partition})). 
Since $\alpha$ is a coloring of $G$, $\alpha\restr_{T}$ is a coloring of $G[T]$, so we have created one node $x_c$ with $c=\alpha(v)$. This is the unique node of $\Cnew$ with $\alpha\in S'_{x_c}$. 

\item 
Consider a node $x_c$ of $\Cnew$, with label $\ell'(x_c)=\delta$. For every $\alpha\in S'_{x_c}$, it holds that $\alpha(v)=c$ and $\delta(v)=c$. Furthermore, $\delta\restr_{T\bs \{v\}}=\ell(x)=\alpha\restr_{T\bs \{v\}}$ (Lemma~\ref{lem:characterizationCSGs}(\ref{pr:correctlabels})). This shows that the label $\ell'(x_c)$ is chosen correctly. 

\item 
Consider two adjacent nodes $x_c$ and $y_d$ of $\Cnew$. If $x=y$ then $c\not=d$, so $\ell'(x)\not=\ell'(y)$. Otherwise, $x$ and $y$ are adjacent nodes in $\Corig$, so $\ell(x)\not=\ell(y)$ (Lemma~\ref{lem:characterizationCSGs}(\ref{pr:propercoloring})). The labels $\ell(x)$ and $\ell(y)$ are the restrictions of $\ell'(x_c)$ and $\ell'(y_d)$ to $T\bs \{v\}$, so also in this case we conclude that $\ell'(x)\not=\ell'(y)$. 

\item 
Consider a node $x_c$ of $\Cnew$, and two $k$-colorings $\alpha$ and $\beta$ in $S'_{x_c}$. There is a path $P$ from $\alpha\restr_{V(G)\bs \{v\}}$ to $\beta\restr_{V(G)\bs \{v\}}$ in the subgraph of $\C_k(G-v)$ induced by $S_x$ (Lemma~\ref{lem:characterizationCSGs}(\ref{pr:connectedsets})). As $N(v) \subseteq T$, all colorings $\gamma$ in $P$ have an {\em extension} $\gamma'\in S_{x_c}$ with $\gamma'(v)=c$ and $\gamma'\restr_{V(G)\bs \{v\}}=\gamma$. 
So replacing all colorings in $P$ by their extension this way yields a path from $\alpha$ to $\beta$ in the subgraph of $\C_k(G)$ induced by $S'_{x_c}$. Therefore, $S'_{x_c}$ is connected.

\item 
Consider two distinct nodes $x_c$ and $y_d$ in $\Cnew$, and their corresponding sets of solutions $S'_{x_c}$ and $S'_{y_d}$. Observe that these contain solutions that are adjacent in $\C_k(G)$ if and only if at least one of the following is true: (1) $c=d$ (and thus $x\not=y$) and $S_x$ and $S_y$ contain solutions that are adjacent in $\C_k(G-v)$, or (2) $c\not=d$ and $S_x\cap S_y\not=\emptyset$. The first case holds if and only if $c=d$ and the nodes $x$ and $y$ are adjacent in $\Corig$ (Lemma~\ref{lem:characterizationCSGs}(\ref{pr:adjacency})). In the second case, $S_x\cap S_y\not=\emptyset$ holds if and only if $S_x=S_y$, and thus $x=y$ (Lemma~\ref{lem:characterizationCSGs}(\ref{pr:partition})). 
This shows that we have added the edges correctly.\qed
\end{enumerate}
\end{proof}

\def\Cnew{H}	
\def\Cone{H_1}	
\def\Ctwo{H_2}	

\medskip
\noindent
{\bf Lemma~\ref{lem:recolJoin} (Join) [Restated].}
{\it Let $(G,T)$ be a terminal graph that is the join of terminal graphs $(G_1,T)$ and $(G_2,T)$.
Let $(\Cone,\ell_1)=C^c_k(G_1,T)$ and $(\Ctwo,\ell_2)=C^c_k(G_2,T)$.
Then $(\Cnew,\ell)=C^c_k(G,T)$ can be computed as follows:
\begin{itemize}
\item For every pair of nodes $x\in V(\Cone)$
 and $y\in V(\Ctwo)$: if $\ell_1(x)=\ell_2(y)$ then introduce a node $(x,y)$ with $\ell((x,y))=\ell_1(x)$. 
\item 
 For two distinct nodes $(x,y)$ and $(x',y')$, add an edge between them if and only if $xx'$ is an edge in $\Cone$ and $yy'$ is an edge in $\Ctwo$.
\end{itemize}
Moreover, for every $k$-coloring $\gamma$ of $G$, if $x$ is the $\gamma\restr_{V(G_1)}$-node in
$\Cone$ and $y$ is the $\gamma\restr_{V(G_2)}$-node in
$\Ctwo$, then $(x,y)$ is the $\gamma$-node in $\Cnew$.}

\begin{proof}
Denote $V_1=V(G_1)$ and $V_2=V(G_2)$. 
For every node~$x$ of $\Cone$, let $S^1_x$ denote the set of $k$-colorings of $G_1$ such that these sets satisfy the properties stated in Lemma~\ref{lem:characterizationCSGs}. Similarly, we define the sets $S^2_x$ for every node~$x$ of $\Ctwo$. In addition, we assume again that these sets coincide with the choices of $\gamma\restr_{V_1}$-nodes and $\gamma\restr_{V_2}$-nodes. 

We define a certificate $S$ for $(\Cnew,\ell)$ as follows. 
For every node $(x,y)$ of $\Cnew$, we define the set $S_{(x,y)}$ to consist of all $k$-color assignments $\alpha$ of $G$ such that $\alpha\restr_{V_1}\in S^1_x$ and $\alpha\restr_{V_2}\in S^2_y$.
For any $k$-coloring $\gamma$ of $G$ and node $(x,y)$ of $\Cnew$, we choose $(x,y)$ to be the $\gamma$-node of $\Cnew$ if and only if $x$ is the $\gamma\restr_{V_1}$-node of $\Cone$ and $y$ is the $\gamma\restr_{V_2}$-node of $\Ctwo$. This obviously guarantees that $\gamma\in S_{(x,y)}$ for the chosen $\gamma$-node $(x,y)$.
To prove the statement, it only remains to show that the new solution sets $S_{(x,y)}$ satisfy the five properties stated in Lemma~\ref{lem:characterizationCSGs}.
\begin{enumerate}[1.]
\item 
First, we show that for every node $(x,y)$ of $\Cnew$, $S_{(x,y)}$ is a nonempty set of $k$-colorings of $G$. The set $S^1_x$ contains at least one coloring $\alpha_1$ of $G_1$, and $S^2_y$ contains at least one coloring $\alpha_2$ of $G_2$ (Lemma~\ref{lem:characterizationCSGs}(\ref{pr:partition})). Both of these colorings yield the coloring $\ell((x,y))=\ell_1(x)=\ell_2(y)$ when restricted to $T$ (Lemma~\ref{lem:characterizationCSGs}(\ref{pr:correctlabels})), so they can be combined into a $k$-color assignment $\alpha$ for $G$. Since all edges of $G$ are part of $G_1$ or $G_2$ (by definition of the join operation), the resulting $\alpha$ is a $k$-coloring of $G$. 

To prove that the sets $S_{(x,y)}$ partition the $k$-colorings of $G$, it now suffices to show that every $k$-coloring $\alpha$ of $G$ is included in exactly one set $S_{(x,y)}$. Consider $\alpha_i=\alpha\restr_{V_i}$ for $i=1,2$. Then $\alpha_1\in S_x$ for exactly one node $x$ of $\Cone$, and $\alpha_2\in S_y$ for exactly one node $y$ of $\Ctwo$ (Lemma~\ref{lem:characterizationCSGs}(\ref{pr:partition})). These nodes have $\ell_1(x)=\alpha\restr_{T}$ and $\ell_2(y)=\alpha\restr_{T}$ (Lemma~\ref{lem:characterizationCSGs}(\ref{pr:correctlabels})), so we have created exactly one node $(x,y)$ with $\alpha\in S_{(x,y)}$. 

\item 
Consider a node $(x,y)$ of $\Cnew$, and a solution $\alpha\in S_{(x,y)}$. Let $\alpha_1=\alpha\restr_{V_1}$. Then $\alpha\restr_{T}=\alpha_1\restr_{T}=\ell_1(x)=\ell((x,y))$ (Lemma~\ref{lem:characterizationCSGs}(\ref{pr:correctlabels})).

\item 
Consider adjacent nodes $(x,y)$ and $(x',y')$ of $\Cnew$. Then by definition, $x$ and $x'$ are adjacent in $\Cone$, so $\ell_1(x)\not=\ell_2(x')$ (Lemma~\ref{lem:characterizationCSGs}(\ref{pr:propercoloring})), and thus $\ell((x,y))\not=\ell((x',y'))$. 

\item 
Consider a node $(x,y)$ of $\Cnew$. We prove that $S_{(x,y)}$ is a connected set in $\C_k(G)$. Consider any two colorings $\alpha,\beta\in S_{(x,y)}$. Define $\alpha_i=\alpha\restr_{V_i}$ and $\beta_i=\beta\restr_{V_i}$ for $i=1,2$. 
Then for $i=1,2$, there exists a path $P^i$ (or {\em recoloring sequence}) from $\alpha_i$ to $\beta_i$, in the subgraph of $\C_k(G_i)$ induced by $S^1_x$ resp.\ $S^2_y$ (Lemma~\ref{lem:characterizationCSGs}(\ref{pr:connectedsets})). All colorings $\gamma$ in both paths satisfy $\gamma\restr_{T}=\ell((x,y)=\ell_1(x)=\ell_2(y)$ (Lemma~\ref{lem:characterizationCSGs}(\ref{pr:correctlabels})). Therefore, we can construct a recoloring sequence from $\alpha$ to $\beta$ that contains only colorings in $S_{(x,y)}$ by first recoloring vertices of $V_1\bs T$ as prescribed by the recoloring sequence $P^1$ (which yields a coloring~$\delta$ of $G$ with $\delta\restr_{V_1}=\beta_1$ and $\delta\restr_{V_2}=\alpha_2$), 
and subsequently recoloring vertices of $V_2\bs T$ as prescribed by the recoloring sequence $P^2$
(which yields the coloring~$\beta$). 
This can be done because $V_1\cap V_2=T$ and neither $P^1$ nor $P^2$ recolors a vertex of $T$. All of the color assignments in the resulting sequence are part of $S_{(x,y)}$ by definition (and they are in fact colorings, as argued above 
in~(\ref{pr:partition})). 

\item 
Consider two distinct nodes $(x,y)$ and $(x',y')$ in $\Cnew$. We prove that they are adjacent if and only if there exist solutions $\alpha\in S_{(x,y)}$ and $\beta\in S_{(x',y')}$ that are adjacent in $\C_k(G)$. 

Suppose $(x,y)$ and $(x',y')$ are adjacent. By definition, this means that $x$ and $x'$ are adjacent (and thus distinct) nodes of $\Cone$, and $y$ and $y'$ are adjacent nodes of $\Ctwo$. So we can choose solutions $\alpha^1\in S^1_x$ and $\beta^1\in S^1_{x'}$ that are adjacent in $\C_k(G_1)$, and solutions $\alpha^2\in S^2_y$ and $\beta^2\in S^2_{y'}$ that are adjacent in $\C_k(G_2)$ (Lemma~\ref{lem:characterizationCSGs}(\ref{pr:adjacency})). Since $\ell_1(x)\not=\ell_1(x')$ (Lemma~\ref{lem:characterizationCSGs}(\ref{pr:propercoloring})), and $\alpha^1\restr_{T}=\ell_1(x)$ and $\beta^1\restr_{T}=\ell_2(x')$ (Lemma~\ref{lem:characterizationCSGs}(\ref{pr:correctlabels})), 
the colorings $\alpha^1$ and $\beta^1$ differ on $T$, and therefore, since they are adjacent, only on $T$ (so their restrictions to $V_1\bs T$ are the same). Similarly, the colorings $\alpha^2$ and $\beta^2$ differ only on $T$.
By definition of $(x,y)$, $\ell_1(x)=\ell_2(y)$, so we can choose a $k$-coloring $\alpha$ of $G$ with $\alpha\restr_{V_1}=\alpha^1$ and $\alpha\restr_{V_2}=\alpha^2$. Similarly, we can choose a $k$-coloring $\beta$ of $G$ with $\beta\restr_{V_1}=\beta^1$ and $\beta\restr_{V_2}=\beta^2$. As argued above, the colorings $\alpha$ and $\beta$ differ only on one vertex in $T$, so they are adjacent in $\C_k(G)$. By their construction, $\alpha\in S_{(x,y)}$ and $\beta\in S_{(x',y'})$, so this proves the first direction.

For the converse, suppose that there exist adjacent colorings $\alpha\in S_{(x,y)}$ and $\beta\in S_{(x',y'})$. Let $\alpha^i=\alpha\restr_{V_i}$ and $\beta^i=\beta\restr_{V_i}$ for $i=1,2$. 
Let $w$ be the (unique) vertex of $G$ with $\alpha(w)\not=\beta(w)$. 
If $w\in T$ then, by using similar arguments as in the previous paragraph, one can verify that $x$ and $x'$ are adjacent nodes in $\Cone$ and $y$ and $y'$ are adjacent nodes in $\Ctwo$, and therefore $(x,y)$ and $(x',y')$ are adjacent in $\Cnew$. We conclude the proof by showing that $w\in T$ always holds. Suppose for contradiction that $w\not\in T$; 
without loss of generality
assume that $w\in V_1\bs T$. So $\alpha^2=\beta^2$. Therefore $S^2_y\cap S^2_{y'}\not=\emptyset$, and thus $y=y'$ (Lemma~\ref{lem:characterizationCSGs}(\ref{pr:partition})). It follows that $\ell_1(x)=\ell_2(y)=\ell_2(y')=\ell_1(x')$. In addition, since $(x,y)$ and $(x',y')$ are distinct nodes 
and $y=y'$, it follows that
$x\not=x'$. 
But $\alpha^1\in S^1_x$ and $\beta^1\in S^2_{x'}$ are adjacent, so $xx'\in E(\Cone)$ (Lemma~\ref{lem:characterizationCSGs}(\ref{pr:adjacency})). This is a contradiction to the fact that
$\ell_1(x)\neq \ell_1(x')$ must hold due to
Lemma~\ref{lem:characterizationCSGs}(\ref{pr:propercoloring}).
\qed
\end{enumerate}
\end{proof}

\end{document}